\documentclass[final]{siamart1116}

\title{Dissipativity of system abstractions obtained using approximate input-output simulation}
\author{Etika~Agarwal%
	\thanks{Department of Electrical Engineering, University of Notre Dame, Notre Dame, IN (\email{eagarwal@nd.edu}, \email{antsaklis.1@nd.edu}, \email{vgupta2@nd.edu})}%
	\and
	Shravan~Sajja%
	\thanks{IBM Research, Bangalore, India
			(\email{ssuryashravankumar@gmail.com})}
	\and
	 Panos J. Antsaklis%
	 \footnotemark[1]
	\and
	Vijay Gupta
	\footnotemark[1]
}

\usepackage{amsfonts}
\usepackage{bm}
\usepackage{amsmath}
\usepackage{esvect}
\newsiamthm{assumption}{Assumption}
\newsiamthm{remark}{Remark}
\usepackage{tkz-graph}
\usetikzlibrary{arrows,automata}
\usepackage{color}
\newcommand{\highlighttext}[1] {\textcolor{black}{#1}}

\begin{document}

\maketitle

\begin{abstract}
This work focuses on the invariance of important properties between continuous and discrete models of systems which can be useful in the control design of large-scale systems and their software implementations. In particular, this paper discusses the relationships between the QSR dissipativity of a continuous state dynamical system and of its abstractions obtained through approximate input-output simulation relations. First, conditions to guarantee the dissipativity of the continuous system from its abstractions are provided. The reverse problem of determining the Q, S and R dissipativity matrices of the abstract system from that of the continuous system is also considered. Results characterizing the change in the dissipativity matrices are provided when the system abstraction is obtained. Since, under certain conditions, QSR dissipative systems are known to be stable, the results of this paper can be used to construct stable system abstractions as well. In the second part of this paper, we analyze the dissipativity of the approximate feedback composition of a continuous dynamical system and a discrete controller. We present illustrative examples to demonstrate the results of this paper.
\end{abstract}

\begin{keywords}
  Dissipativity, passivity, abstraction, simulation
\end{keywords}
\section{Introduction} 

Discrete event and hybrid system models for continuous systems are quite common. For example, such models are useful for sampled and quantized systems, and also in software implementations of continuous systems. Furthermore, such discrete models can be very useful in the control design of large dynamical systems, especially when there are temporal logic performance specifications and verification of safety requirements, e.g., in robotic systems. 
This is the main motivation behind studying system properties of interest that are present in both continuous and discrete models of the same system.

These discrete models can be obtained using abstraction based approaches. See for example, \cite{girard2}-\cite{zamani2012}. An abstracted system model approximates a continuous state dynamical system by a system with a pre-order or equivalence relation between the two systems. 
 Control design of dynamical systems using abstraction based approaches can be carried out efficiently based on two main factors \cite{girard3}, first, the possibility of constructing symbolic or purely discrete abstractions of the original system and second, the possibility to infer the behavior of the given continuous system based upon its discrete abstraction. The idea of using discrete abstractions for control design is motivated by the fact that control of a continuous system with a discrete controller requires the use of a continuous to discrete and discrete to continuous conversion. This set-up can be viewed as an interconnection of a continuous system with a software system, as shown in \cref{fig:sys_soft}. 



Dissipativity based approaches provide attractive alternatives to control analysis and design of such systems. These dissipativity approaches have been used for systems with control performance described in terms of stability and optimality requirements \cite{etika2}\cite{arash1} for continuous systems. Dissipativity is an energy based input-output property of dynamical systems \cite{willems}. 
A special form of dissipativity is QSR dissipativity which allows transparent relationships with several important concepts such as passivity and \(\mathcal{L}_2\) stability \cite{xia1}-\cite{moylan}. Further, QSR dissipativity is preserved over feedback and parallel interconnections; and series interconnections under certain conditions \cite{bao}. Hence, an added advantage of using control analysis and design techniques based on QSR dissipativity is that they scale well.

More recently, \cite{arcak}\cite{arcak2} used dissipativity like concepts for compositional analysis of interconnected systems with safety and temporal logic specifications. This opens up a new application area for dissipativity theory where controllers can be designed to meet temporal logic constraints (such as safety and reachability constraints) together with traditional specifications such as passivity, stability and optimality. As a first step towards this goal, in this paper, we analyze the dissipativity of system abstractions and find relationships between the QSR dissipativity, and hence passivity of systems and their approximately input-output similar  abstractions \cite{tazaki-hscc-2008}\cite{tazaki}.



\begin{figure}[h]
	\vspace{-2mm}
	\centering
	\includegraphics[scale=0.4]{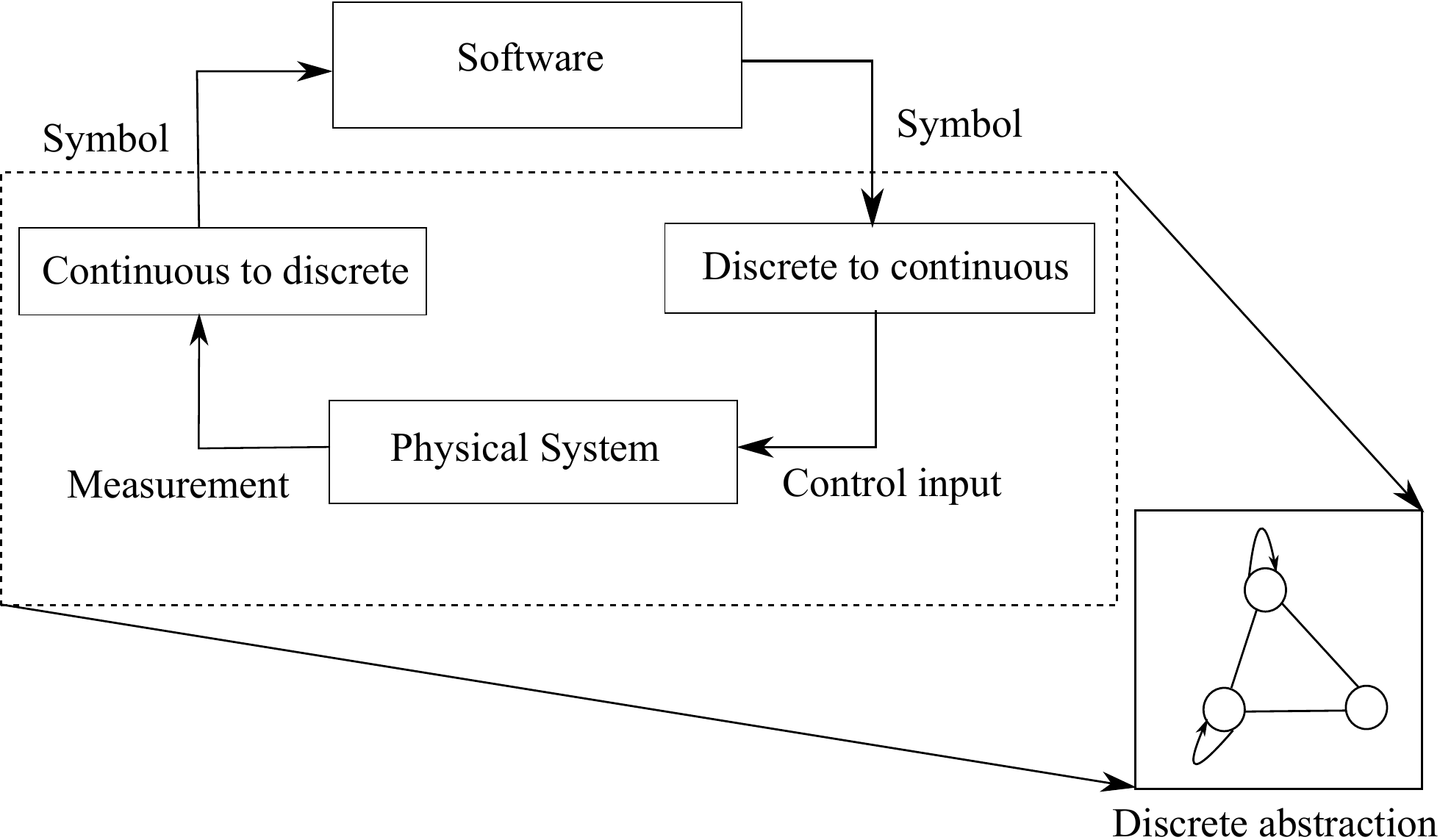}
	\caption{Finite state approximation of a continuous-time plant
		interacting with a finite state controller (software).}
	\label{fig:sys_soft}
\end{figure}

There are several approaches to system abstraction. Different techniques preserve different properties of the system such as reachability \cite{alur1} and compositionality \cite{girard1} while maintaining equivalence or pre-order in a certain sense between the two systems. One of the popular approaches for abstraction is using the notions of simulation, bisimulation and their approximate versions \cite{alur1}\cite{girard2} which have been used for both discrete and continuous time systems \cite{girard2}\cite{van der schaft}\cite{pappas}\cite{zamani2012}. \highlighttext{While these notions are composition preserving in some sense \cite{girard2}, approximate simulation and bisimulation are not preserved for more general input-output interconnections such as cascade and feedback. Tazaki et al \cite{tazaki-hscc-2008} proposed a modification of the definitions of approximate simulation and bisimulation which, under mild conditions, are interconnection preserving.} This means that the abstraction of interconnected system is the interconnection of system abstractions. However, \cite{cujipers} identified that stability is not preserved for simulation (or bisimulation) and additional continuity constraints need to be imposed on the simulation relation. Prabhakar et al \cite{prabhakar}\cite{prabhakar2} also proposed some continuity constraints on simulation relation to preserve variants of stability during abstraction. An alternate approach is offered in the present paper using dissipativity. In this work, we are interested in inferring QSR dissipativity properties of discrete abstractions from that of the corresponding continuous system. These discrete abstractions are obtained using a slight modification of method in \cite{zamani2012} and are related to the continuous system through approximate input-output simulation relation. \highlighttext{The definition of approximate input-output simulation relation used in this paper is a special case of the interconnection-compatible approximate simulation relation in \cite{tazaki-hscc-2008}}. QSR dissipative systems, under mild additional conditions, are stable. This allows us to obtain stable and composable system abstractions without imposing any strict constraints on the simulation relation. Related work on the passivity of bisimilar systems was done in \cite{xu_passivity_paper}, where the problem of discretization of a continuous controller while preserving the passivity of closed loop system was considered. Note that while \cite{xu_passivity_paper} discusses the passivity indices of approximate bisimulation of the controller, in the present paper, we provide the relationship between the more general QSR dissipativity of systems associated through approximate input-output simulation relations. We additionally discuss the QSR dissipativity of an approximate feedback composition \cite{tabuada2009} of systems. Moreover, as opposed to \cite{xu_passivity_paper}, we do not restrict the systems to be incrementally stable. Some parts of our work are contained in \cite{shravan}, but have not been published previously.

The main contribution of this paper is to provide a relationship between the QSR dissipativity of a system and its approximately input-output similar abstraction. Briefly, the results of this paper enable us to determine the Q, S and R dissipativity matrices of an abstract system by analyzing the dissipativity of the original continuous system and vice-versa. As a special case, these results also relate the passivity levels and stability of two systems associated through an approximate input-output simulation relation. We limit ourselves to the class of systems and abstract models which allow for a definition of inner product between inputs and outputs of system. An example of this is the abstraction used in \cite{zamani2012}. This allows us to apply the standard dissipativity definitions to abstractions of systems as well. 

 The concept of approximate input-output simulation relation is defined for both the continuous and discrete system abstractions. It is possible to obtain a continuous as well as discrete abstraction of the same dynamical system. The first result of this paper, provides conditions under which the QSR dissipativity of an approximately input-output similar continuous (continuous state) abstraction, implies the dissipativity of the original continuous dynamical system. We further consider the reverse problem of characterizing the dissipativity of a discrete state abstraction when the given continuous system is QSR dissipative. Two different cases are presented, (i) when the system state is the measured output, and (ii) when system state is not measured, instead measured output is same as the system output. The main difference in these two cases arises from the fact that the approximate input-output simulation relation between a system and its abstraction holds with respect to measured output. In both of these cases, the class of systems are incrementally forward complete and the abstraction is obtained using the approach in \cite{zamani2012}. Since both the approximate input-output simulation and QSR dissipativity are composition preserving, composition of QSR dissipativities of abstractions is the same as the QSR dissipativity of the abstraction of composition. 

In the second part of this paper we analyze the QSR dissipativity of composition of continuous state systems with discrete state systems, say software. Specifically, we use the approximate feedback composition defined in \cite{tabuada2009}. We show that once two systems are approximately feedback composable, then QSR dissipativity of even one of those systems implies QSR dissipativity of the composition.

This paper is organized as follows. \Cref{background} describes the system model and some important dissipativity notions and system relations. \Cref{main result} addresses the relationship between QSR dissipativity of approximately input-output similar systems. 
We also present results to quantify the dissipativity matrices of approximately input-output similar systems for particular abstraction methods. \Cref{example} discusses examples to illustrate the results. In \Cref{composition}, some results on the dissipativity of approximate feedback composition of finite transition systems are provided. Finally, \Cref{conclusion} contains concluding remarks.

\highlighttext{
\begin{table}[tbhp]
\caption{\highlighttext{\bf Notation}}
\label{tab:notation}
\centering
\begin{tabular}{|c|c|} \hline
\highlighttext{\bf Notation} & \highlighttext{\bf Meaning}  \\ \hline
$\| \cdot \| $ & infinity norm \\ 
$\| \cdot \|_2 $ & Euclidean norm (or induced 2-norm if argument is a matrix) \\
$\mathbb{Z}_0^+$ & set of nonnegative integers\\
$\mathbb{R}$ & set of real numbers\\
$\mathbb{R}^+$ & set of positive real numbers\\
$\mathbb{R}^+_0$ & set of nonnegative real numbers\\
$\mathbb{R}^n$ & Euclidean space of dimension $n$\\
$\vv{0}$ & zero vector of appropriate dimension \\
$P>Q$ & matrix $(P-Q)$ is positive definite 
\\\hline
\end{tabular}
\end{table}
}

\section{Preliminaries}
\label{background}
In this section we briefly explain few important notions of abstract systems and introduce the properties of dissipativity, passivity and incremental forward completeness. 

\subsection{System Description}
The system model we use in this work is that of transition systems. The following definitions are standard and can be found in \cite{zamani2012}.

\begin{definition} A transition system \(T=(X,U,\longrightarrow,Y_m,\mathcal{H})\) consists of:
	\begin{itemize}
		\item a set of states X;
		\item a set of inputs U;
		\item a transition map \(\longrightarrow:X\times U\rightarrow X\);
		\item a set of outputs \(Y\);
		\item and an output map \(\mathcal{H}:X\times U\rightarrow Y\).
	\end{itemize}
\end{definition}
If for any state $x\in X$ and $u \in U$ there exists at most one state $x' \in X$ such that $x \xrightarrow [ \quad  \quad]{u}x'$ then the system is deterministic. $x'$ is also known as the $u$-successor of $x$. If the system is nondeterministic, then for a transition $x \xrightarrow [ \quad  \quad]{u}x'$ the state $x'$ may not be unique.  In such a case $x'$ belongs to a set of all possible $u$-successors given by ${Post}_u(x)$ and we will use $U(x)$ to denote the set of inputs $u \in U$ for which  ${Post}_u(x)$ is nonempty.
Suppose the transition system is equipped with metrics $\mathbf{d_X}: X\times X \longrightarrow \mathbb{R}_0^+$, $\mathbf{d_U}: U\times U \longrightarrow \mathbb{R}_0^+$ and $\mathbf{d_Y}: Y\times Y \longrightarrow \mathbb{R}_0^+$ representing the ``distance" between two elements of state space, input space and output space respectively. This transition system is referred to as a \textbf{metric transition system}. \highlighttext{Throughout this work, we assume that the transition system allows for a notion of inner product between inputs and outputs and the distance metric $\mathbf{d_X}$, $\mathbf{d_U}$ and $\mathbf{d_Y}$ are infinity norms $\|\cdot\|$.} 

Transition systems can be used to describe a large class of dynamical systems. We restrict ourselves to continuous time dynamical systems of the form \begin{equation} \label{sigma}
	\Sigma=(X,U,Y_m,f,h_m)
\end{equation}
where \highlighttext{$X = \mathbb{R}^n$ is the state space; $U\subseteq \mathbb{R}^m:\{\vv{0}\} \in U$ is the input space; $Y_m \subseteq \mathbb{R}^p: \{\vv{0}\}\in Y_m$ is the measured output space; $f: X \times U \to X$ is a Lipschitz continuous map describing state transition and $h_m: X \times U \to Y_m$ is the measured output map. At at any time $t\in \mathbb{R}_0^+$, the state, input and measured output of $\Sigma$ are $x(t) \in X$, $u(t) \in U$, $y_m(t) \in Y_m$ and the state and measured output evolve as $\dot{x}(t) = f(x(t),u(t))$ and $y_m(t) = h_m(x(t),u(t))$.} If $\xi: ]a, b[ \xrightarrow [ \quad]{} X$ is a solution  of the differential equation $\dot{x}(t)=f(x(t), u(t))$, then we will use $\xi( t, x, u)$ to denote a unique point reached at time $t$ under the input signal $u:[0,t]\rightarrow U$ from an initial condition $x \in X$. The transition system associated with $\Sigma$ is then given by $T(\Sigma)=(
X,U,\longrightarrow,Y_m,h_m)$ where the state transition map is dictated by the differential equation $\dot{x}=f(x, u)$. 
We use notation $T(\Sigma)$ and $T$ interchangeably in this work. Also, \highlighttext{for ease of notation, we will often drop the time index when referring to the state, input and output of $\Sigma$.}

For dissipativity analysis, we consider a separate system output dictated by system output space $Y\subseteq \mathbb{R}^m$ and output map $h(x(t),u(t)):X\times U \to Y$. The measured output $h_m(x(t),u(t))$ can be different from system output $h(x(t),u(t))$. In this work, we consider two different cases (i) when measured output is the system states, and (ii) when measured output is same as the system output. The measured output space $Y_m$ and output map $h_m$ take values accordingly.

We can also define a discrete time system $\Sigma_d=(X_d,U_d,Y_{m_d},f_d,h_{m_d})$ where $X_d,\ U_d,$ $Y_{m_d}, \ f_d$ and $h_{m_d}$ are state space, input space, measured output (or measurement) space, state transition map and measured output transition map respectively. At any discrete time $k$ the system state $x_d(k)\in X_d$, input $u_d(k)\in U_d$ and output $y_{m_d}(k)\in Y_{m_d}$ evolve in discrete time steps as $x_d(k+1)=f_d(x_d(k),u_d(k))$ and $y_{m_d}(k) = h_{m_d}(x_d(k),u_d(k))$ for all $k\in \mathbb{Z}_0^+$. Similar to continuous time case, system output (used for dissipativity)  dictated by system output space $Y_d \subseteq \mathbb{R}^m$ and output map $h_d(x_d(k),u_d(k)):X_d\times U_d \to Y_d$, can be different from measured output.

The following assumptions on system behavior are useful in deriving the main results of this paper. 

\begin{assumption}\label{incrementally forward complete}(Incremental forward completeness \cite{zamani2012}) The dynamical system \(\Sigma\) is said to be incrementally forward complete if there exist continuous functions \(\alpha_1 : \mathbb{R}^+_0 \times \mathbb{R}^+_0 \rightarrow \mathbb{R}^+_0\), and $\alpha_2: \mathbb{R}_0^+ \times \mathbb{R}_0^+ \to \mathbb{R}_0^+$, \(\alpha_1(\cdot,t), \alpha_2(\cdot,t)\in \mathcal{K}_{\infty}\) for every \(t \geq 0\), such that for any two initial conditions \(x_1,x_2 \in X\), any input trajectories $v_1,v_2: [0,t] \in U$, and for any \(t\in \mathbb{R}^+_0\) the following bound holds: 
	\begin{equation}\label{beta_gamma}
	\|\xi( t, x_1, v_1)-\xi( t, x_2, v_2)\|\leq \alpha_1(\|x_1-x_2\|,t)+\alpha_2(\|v_1-v_2\|,t)
	\end{equation}
	where \(\xi(t,x_i, v_i)\) is the state trajectory of system with input \(v_i\) and initial state \(x_i\). 
\end{assumption}

It should be noted that incremental forward completeness is a weaker condition than incremental stability for it does not require the system to be stable. It only states that the distance between any two state trajectories is bounded.

\begin{assumption}\label{output rate gain bound} 
	Assume that the operator from input $u(t)$ to rate of change of system output $\dot{y}(t)$ has the finite $\mathcal{L}_2$ gain $\gamma$, that is 
	$$
	\int_0^{\tau}\|\dot{y}(t)\|^2_2\text{d}t\leq \gamma^2 \int_0^{\tau}\|u(t)\|^2_2\text{d}t
	$$
	for any ${\tau}\geq 0$ and admissible input $u(t)$.
\end{assumption}

Assumption \ref{output rate gain bound} is an $\mathcal{L}_2$ gain condition which bounds the rate at which the output $y$ can change with respect to time.

We define transition systems $ T_{\tau}(\Sigma) $ obtained after sampling $\Sigma$, and $T_{\tau,\mu,\eta}(\Sigma)$ obtained after appropriate sampling and quantization of $\Sigma$ as follows.
\begin{definition} \cite{zamani2012}\label{def:discrete-time sys} 
	Let $\Sigma$  be a dynamical system and the associated transition system be $T(\Sigma)$. \hspace{-1mm} For any $\tau >0$, the sampled transition
	system $ T_{\tau}(\Sigma)\hspace{-1mm}:=\hspace{-1mm}(X_{\tau},U_{\tau},  \xrightarrow [\tau]{u_{\tau}},$ $Y_{m_{\tau}},\mathcal{H}_{m_{\tau}})$ is defined by:
	\begin{itemize}
		\item $X_{\tau}=X$;
		\item $U_{\tau}= U$;
		\item $x_\tau \xrightarrow [\quad \tau \quad]{u_{\tau}} x_{\tau}'$, if there exists a trajectory $\xi: [0, \tau]\xrightarrow [\quad \quad]{}$ $\xi(\tau, x_\tau,u)=x_{\tau}'$ where $u:[0,\tau)\to{u_\tau}$, $u_\tau\in U_\tau$;
		\item $Y_{m_{\tau}}=Y_m$;
		\item \highlighttext{$\mathcal{H}_{m_\tau}(x_\tau,u_\tau)=h_m(x_\tau,u_\tau)$ where $x_\tau\in X_\tau, u_\tau\in U_\tau$.}
	\end{itemize}
\end{definition}
$T_\tau(\Sigma)$ evolves in discrete time and can be represented as a discrete time system (like $\Sigma_d$). As discussed before, for dissipativity analysis, we consider another output (referred to as system output in this work) associated with $T_\tau(\Sigma)$. At any discrete time step $k$, system output is described by the output space $Y \subseteq \mathbb{R}^m$ and output function $h(x(k),u(k))=y(k) \in Y$ where $x(k)\in X_\tau$ and $u(k) \in U_\tau$.

\begin{definition} \label{sampled-quantized system different output} For any incrementally forward complete control system $\Sigma$, with system states as measurement, i.e., $h_m(x,u)=x$, and parameters $\tau >0$, $\eta >0$, $\mu >0$ and design parameters $\theta_1, \theta_2 \in \mathbb{R}^+$,  a countable transition system  can be defined as, $\label{T_tau_mu_eta1} T_{\tau, \mu, \eta}(\Sigma):=(X_{\textbf{q}},U_{\textbf{q}}, \xrightarrow [\quad \tau \quad]{u_{\textbf{q}}}, Y_{m_q}, \mathcal{H}_{m_q} )$, where:
	\begin{itemize}
		\item \(X_{\textbf{q}}=[X]_{\eta}=\begin{Bmatrix}
		x\in X| x_i = k_i\eta, k_i \in \mathbb{Z},\)  and \( i = 1,2,\dots, n
		\end{Bmatrix}\);
		\item $U_{\textbf{q}}=[U]_{\mu}=\begin{Bmatrix}
	u\in U| u_i = k_i\mu, k_i \in \mathbb{Z}, \text{ and } i = 1,2, \dots, m
		\end{Bmatrix}$;
		\item $x_{\textbf{q}} \xrightarrow [\quad \tau \quad]{u_{\textbf{q}}} x_{\textbf{q}}'$, if  $\|\xi(\tau, x_{\textbf{q}}, u_{\textbf{q}})-x_{\textbf{q}}'\|\leq \alpha_1(\theta_1, \tau)+\alpha_2(\theta_2, \tau)+\eta/2$;
		\item $Y_{m_q}=[X]_{\eta}=\begin{Bmatrix}
		x\in X| x_i = k_i\eta, k_i \in \mathbb{Z}, \text{ and } i = 1,2, \dots, n
		\end{Bmatrix};$
		\item \highlighttext{$\mathcal{H}_{m_q}(x_q,u_q)=x_q$ where $x_q \in X_q, u_q \in U_q$.} 
	\end{itemize}
	$\alpha_1$ and $\alpha_2$ are functions from the definition of incremental forward completeness in \Cref{incrementally forward complete}. 
\end{definition}

Here, \highlighttext{\(\mathcal{H}_{m_q}(x_q,u_q)=x_q\) indicates that measured output is same as the system states. At any discrete time instant $k\in \mathbb{Z}^+_0$ system output (used for dissipativity analysis) of $T_{\tau, \mu, \eta}(\Sigma)$ here is \(h(x_q(k),u_q(k))\) where $x_q(k)\in X_q$ and $u_q(k) \in U_q$.} Note that this set up is often useful when analyzing systems with state feedback.

We also define the sampled and quantized transition system for the case when system output is the measurement.

\begin{definition} \label{sampled-quantized system same output} For any incrementally forward complete control system $\Sigma$, with system output as measurement, i.e., $h_m(x,u)=h(x,u)$ and parameters $\tau >0$, $\eta >0$, $\mu >0$ and design parameters $\theta_1, \theta_2 \in \mathbb{R}^+$,  a countable transition system  can be defined as $\label{T_tau_mu_eta2} T_{\tau, \mu, \eta}(\Sigma):=(X_{\textbf{q}},U_{\textbf{q}}, \xrightarrow [\quad \tau \quad]{u_{\textbf{q}}}, Y_{m_q}, \mathcal{H}_{m_q} )$,	where:
	\begin{itemize}
		\item $X_{\textbf{q}}=[X]_{\eta}=\begin{Bmatrix}
		x\in X| x_i = k_i\eta, k_i \in \mathbb{Z},\)  and \( i = 1,2,\dots, n
		\end{Bmatrix}\);
		\item $U_{\textbf{q}}=[{U}]_{\mu}=\begin{Bmatrix}
		u\in U| u_i = k_i\mu, k_i \in \mathbb{Z}, \text{ and } i = 1,2, \dots, m
		\end{Bmatrix}$;
		\item $x_{\textbf{q}} \xrightarrow [\quad \tau \quad]{u_{\textbf{q}}} x_{\textbf{q}}'$, if  $\|\xi(\tau, x_{\textbf{q}}, u_{\textbf{q}})-x_{\textbf{q}}'\|\leq \alpha_1(\theta_1, \tau)+\alpha_2(\theta_2, \tau)+\eta/2$;
        \item $Y_{m_q}=[Y]_{\mu}=\begin{Bmatrix}
		y\in Y| y_i = k_i\mu, k_i \in \mathbb{Z}, \text{ and 
		} i = 1,2, \dots, m
		\end{Bmatrix}$;
		\item \highlighttext{$\mathcal{H}_{m_q}(x_q,u_q)= h_q(x_q,u_q): \|h_q(x_q,u_q)-h(x_q,u)\|\leq \mu/2$ where $u:[0,\tau)\to u_q$, $x_q \in X_q, u_q \in U_q.$}
	\end{itemize}
	$\alpha_1$ and $\alpha_2$ are functions from the definition of incremental forward completeness in \Cref{incrementally forward complete}. 
\end{definition}

As compared to \cref{sampled-quantized system different output}, since \Cref{sampled-quantized system same output} has system output  as measurement, system output (used for dissipativity analysis) is also quantized and is described by the output map $\mathcal{H}_{m_q}: X_q\times U_q \to Y_{m_q}$ in \Cref{sampled-quantized system same output}.

The transition system $T_{\tau, \mu, \eta}(\Sigma)$ can be countably finite or infinite depending on the size of state and input spaces. For most practical cases, the system states and inputs are restricted due to the physical limitations of the system leading to a countably finite $T_{\tau, \mu, \eta}(\Sigma)$.

\subsection{System Relations}
Abstraction is an approach to reduce the complexity of the description of dynamical systems. 
One of the popular approaches for abstraction is using the notion of  approximate simulation \cite{girard2}. 
Since dissipativity is an input-output property, we talk about a generalized notion of approximate input-output simulation \cite{tazaki}. 

Consider two metric transition systems $T_1$ and $T_2$. The approximate input - output simulation relation can be defined as follows. 
\begin{definition} \label{def:IOS} \(T_2\) is an approximate input-output simulation of \(T_1\) with precision \((\epsilon_u,\epsilon_y)\) if there exists an approximate input-output simulation relation \(\mathcal{R} \subseteq X_1 \times X_2\) such that for all \(x_1\in X_1\), there exists \(x_2\in X_2\) such that \(({x}_1,{x}_2)\in \mathcal{R}\) and, for all \((x_1,x_2)\in \mathcal{R}\):
	\begin{enumerate}
		\item for all $u_1 \in U_1(x_1)$ there exists $u_2 \in U_2(x_2)$ such that  $\mathbf{d_U}(u_1,u_2)\leq \epsilon_u$ and $\mathbf{d_Y}(\mathcal{H}_1(x_1,u_1), \mathcal{H}_2(x_2,u_2))\leq \epsilon_y$,
		\item for all $u_1 \in U_1(x_1)$ there exists $u_2 \in U_2(x_2)$ such that  $\mathbf{d_U}(u_1,u_2)\leq \epsilon_u$ and    $x_1 \xrightarrow [\quad1 \quad]{u_1}x'_1$ in $T_1$ implies the existence of  $x_2 \xrightarrow [\quad 2 \quad]{u_2}x'_2$ in $T_2$ such that $ (x'_1,x'_2)\in \mathcal{R}$.
	\end{enumerate}
	This is denoted as $T_1\preceq^{(\epsilon_u, \epsilon_y)}_{IOS} T_2$.
	
\end{definition}    

\highlighttext{We can decompose the input and output vectors into two groups, internal and external signals as in \cite{tazaki-hscc-2008} to facilitate the discussion on system interconnection. It is easy to see that in this case, \Cref{def:IOS} becomes a special case of interconnection-compatible approximate simulation in \cite{tazaki-hscc-2008} where not only the internal inputs but also external inputs of the two interconnection-compatible approximately similar systems are required to be close enough to each other. As such, under mild conditions (Theorem 1 in \cite{tazaki-hscc-2008}), we can interconnect abstractions (approximate input-output simulation) of systems  to obtain an abstraction of interconnection of systems.} 

Similarly, we can also define the notion of approximate input-output alternating simulation to take into account the non-deterministic nature of system trajectories. 

\begin{definition} \cite{zamani2012} \label{def:IOAS}
	\(T_2\) is an approximate input-output alternating simulation of \(T_1\) with precision \((\epsilon_u,\epsilon_y)\) if there exists an approximate input-output alternating simulation relation \(\mathcal{R} \subseteq X_1 \times X_2\) such that for all \(x_1\in X_1\), there exists \(x_2\in X_2\) such that \(({x}_1,{x}_2)\in \mathcal{R}\) and, for all \((x_1,x_2)\in \mathcal{R}\):	
	\begin{enumerate}
		\item for all $u_1 \in U_1(x_1)$ there exists $u_2 \in U_2(x_2)$ such that  $\mathbf{d_U}(u_1,u_2)\leq \epsilon_u$ and $\mathbf{d_Y}(\mathcal{H}_1(x_1,u_1), \mathcal{H}_2(x_2,u_2))\leq \epsilon_y$,
		\item for all $u_1 \in U_1(x_1)$ there exists $u_2 \in U_2(x_2)$ such that  $\mathbf{d_U}(u_1,u_2)\leq \epsilon_u$ and for every $x'_2 \in {Post}_{u_2}(x_2)$ there exists $x'_1 \in Post_{u_1}(x_1)$ such that $ (x'_1,x'_2)\in \mathcal{R}$.
	\end{enumerate}
	This is denoted as \(T_1\preceq^{(\epsilon_u, \epsilon_y)}_{IOAS} T_2\).
\end{definition}

\noindent The two notions of alternating approximate input-output simulation and approximate input-output simulation coincide in the special case of deterministic systems.

\subsection{Dissipativity} \label{subsection:dissipativity}

A dynamical system is said to be dissipative if it stores and dissipates energy but does not generate any energy of its own. The notion of energy mentioned here is general and is captured using the energy supply rate function. As described in \cite{willems}, the supply rate function is a real valued function described on the input and output space and is locally integrable, i.e., $\omega: U \times Y \longrightarrow \mathbb{R}, \,\, \int_{t_0}^{t_1}|\omega(u(t),y(t)|dt<\infty \,\, \forall \ t_1,t_0\in\mathbb{R}_0^+$.
\begin{definition}\label{def:dissipativity cont} \cite{willems}
\highlighttext{A continuous time system $\Sigma$ with output function $y(t) = h(x(t),u(t))$, is said to be dissipative with respect to the supply rate function $\omega(u,y)$, if there exists a nonnegative function \highlighttext{$V:X\rightarrow\mathbb{R}^+$}, called the storage function, such that for all $t_1\geq t \geq t_0\geq 0$, $x(t_0)\in X$ and $u(t) \in U$
	\begin{equation}
	\int_{t_0}^{t_1}\omega(u(t),  y(t))dt \geq V(x(t_1))-V(x(t_0))
	\end{equation}
	holds where $x(t_1)$ is the state at time $t_1$ resulting from the initial condition $x(t_0)$ and input function $u(\cdot)$. }
\end{definition}

\Cref{def:dissipativity cont} discusses dissipativity of continuous time system. 
Similarly, we can define dissipativity for discrete time system (or sampled transition system) as follows.

\begin{definition}\cite{byrnes1994}\label{def:dissipativity discrete}
\highlighttext{A discrete time system $\Sigma_d$ with output function $y(k)= h(x(k),u(k))$, is said to be dissipative with respect to the supply rate function $\omega(u,y)$, if there exists a nonnegative function $V:X_d\rightarrow\mathbb{R}^+$, called the storage function, such that for all $k\geq k_0\geq 0$, $x_0\in X_d$ and $u(k) \in U_d$
	\begin{equation} \label{diss inequality discrete1}
	\sum_{i=k_0}^{k-1} \omega(u(i),  y(i)) \geq V(x(k))-V(x(k_0))
	\end{equation}
	holds where $x(k)$ is state at $k$ resulting from the initial condition $x(k_0)$ and input function $u(\cdot)$.}
\end{definition}

We can obtain different dissipativity structures and system properties depending on the choice of supply rate function.
\begin{definition} \label{def:QSR} \highlighttext{(QSR dissipativity and passivity)
\begin{itemize}
\item A continuous time system $\Sigma$ is said to be \(QSR\) dissipative if it is dissipative with respect to the supply rate \(\omega(u,y) = y^TQy + 2 y^TSu + u^TRu\). Dissipativity inequality is then given as
	\begin{equation}\label{QSRcontinuous}
	\int_{t_0}^{t_1} (y^T(t)Qy(t) + 2 y^T(t)Su(t) + u^T(t)Ru(t))dt \geq V(x(t_1))-V(x(t_0))
	\end{equation}
	where \(Q\), \(S\) and \(R\) are matrices of appropriate dimensions.
    \item Similarly, a discrete time system $\Sigma_d$ is said to be \(QSR\) dissipative if it is dissipative with respect to the supply rate \(\omega(u,y) = y^TQy + 2 y^TSu + u^TRu\). Dissipativity inequality is then given as
	\begin{equation}\label{QSRdiscrete}
	\sum_{i=k_0}^{k-1} (y^T(i)Qy(i) + 2 y^T(i)Su(i) + u^T(i)Ru(i)) \geq V(x(k))-V(x(k_0))
	\end{equation}
	where \(Q\), \(S\) and \(R\) are matrices of appropriate dimensions.
    \item \label{passivity}
	A continuous time system $\Sigma$ (or discrete time system $\Sigma_d$) is said to be input feed-forward output feedback passive if it is QSR dissipative with \(Q = -\rho \mathbf{I}\), \(S = \frac{1}{2} \mathbf{I}\), \(R = -\nu \mathbf{I}\) with $\rho, \nu \geq 0$.  
\end{itemize}}  
\end{definition}

In this work, we refer to \(Q, \text{\space}S, \text{ and } R\) matrices as dissipativity matrices. A special case of QSR dissipativity is input feed-forward output feedback passivity with $\nu $ and $\rho$ known as the input and output passivity index. In this work, we use the terms passive and input feed-forward output feedback passive interchangeably. 

\begin{definition} (Quasi-dissipativity \cite{lozano2007}) \label{def:quasi-dissipative}\highlighttext{
A continuous time system $\Sigma$ (or discrete time system $\Sigma_d$) is said to be quasi-dissipative with respect to $\omega(u,y)$ if there exists a constant $\beta \leq 0$ such that it is dissipative with respect to the supply rate $\omega(u,y)-\beta$.}
\end{definition}

The boundedness and stability of quasi-dissipative systems is discussed in \cite{polushin2004}. Similar to dissipativity, quasi-dissipativity also can take different forms depending on the structure of supply rate function $\omega(u,y)$. Quasi-QSR-dissipativity and quasi-passivity are two special cases that are of interest to us. They can be defined similar to QSR dissipativity and passivity described earlier. In this work, dissipativity for systems in \Cref{sampled-quantized system different output} and \Cref{sampled-quantized system same output} is discussed in the context of quasi-dissipativity where, the presence of $\beta$ on the right hand side of \cref{quasi-diss inequality discrete2} indicates the energy generated due to quantization process.

\begin{remark} 
\highlighttext{The definition of dissipativity here is independent of system representation, i.e., if $\Sigma$ is dissipative then $T(\Sigma)$ is also dissipative. Also note that while transition system $T(\Sigma)$ follows dissipativity definition for continuous time system, trajectories of transition systems in \Cref{def:discrete-time sys}, \Cref{sampled-quantized system different output} and \Cref{sampled-quantized system same output} evolve in discrete time and hence follow the dissipativity definition for discrete-time systems.} \highlighttext{ Moreover, \cite{byrnes1994} showed that for discrete time systems, \cref{diss inequality discrete1} holds if and only if
\begin{equation} \label{diss inequality discrete2}
\omega(u(k),y(k)) \geq V(x(k+1))-V(x(k))
\end{equation}
for all $k \in \mathbb{Z}_0^+,\ u(k) \in U$ and $ x(k) \in X$. Equivalent condition for quasi-dissipativity is,
\begin{equation} \label{quasi-diss inequality discrete2}
\omega(u(k),y(k)) \geq V(x(k+1))-V(x(k))+\beta \ \ \forall \ k \in \mathbb{Z}_0^+,\ u(k) \in U_d,\ x(k) \in X_d.
\end{equation}}
\end{remark} 




\section{Dissipativity of systems and their abstractions}
\label{main result}  
In this section we discuss the relationship between the dissipativity properties of a continuous system and its abstraction. 
We provide two main results. First, we analyze the dissipativity of a continuous dynamical system when its approximate input-output simulation is QSR dissipative. Secondly, we consider the reverse problem of determining the dissipativity of a system abstraction. We provide conditions under which QSR dissipativity of a continuous system implies QSR dissipativity of its discrete abstraction obtained using the approach in \cite{zamani2012}.

\subsection{Dissipativity of system from its QSR dissipative abstraction}

Consider two continuous time systems $\Sigma_1$ and $\Sigma_2$ and corresponding transition systems $T_1(\Sigma_1)$ and $T_2(\Sigma_2)$ 
. Let $u_i$, $y_i$ and $x_i$ represent respectively the input, output and states of \(T_i\), $i\in\{1,2\}$ and both these systems allow a notion of inner product between their inputs and outputs. Suppose \(T_2\) is QSR dissipative with \(Q_2\), \(S_2\), \(R_2\) as the dissipation matrices. If \(T_2\) approximately simulates \(T_1\) then the following theorem gives the conditions under which \(T_1\) is \(QSR\) dissipative.
\begin{theorem} \label{result1} \hspace{-2mm}
If \(T_1(\Sigma_1)\preceq^{(\epsilon_u, \epsilon_y)}_{IOS} T_2(\Sigma_2)\) and \(T_2(\Sigma_2)\) is \(QSR\) dissipative, then \(T_1(\Sigma_1)\) is also QSR dissipative with matrices \(Q_1, S_1, R_1\) satisfying
	\begin{gather}
	\label{cond1}
	\begin{aligned}
	\underline{\lambda}(Q_1-Q_2)-\zeta_1\|Q_1\|^2_2 -\zeta_3 &\geq 0\\
	\underline{\lambda}(R_1-R_2)- \zeta_2\|S_1\|^2_2 -\zeta_4 \|R_1\|^2_2 &\geq 0\\
	S_1 &= S_2
	\end{aligned}
	\end{gather}
	where \(\zeta_1\), \(\zeta_2\), \(\zeta_3\), \(\zeta_4\) \(\in \mathbb{R}^+\) are arbitrary non-zero constants, \(Q_2, S_2, R_2\) are the dissipativity matrices for \(T_2(\Sigma_2)\) and \(\underline{\lambda}(\cdot)\) represents the smallest eigen value of the matrix in discussion.
\end{theorem}
\begin{proof}
	See Appendix. 
\end{proof}

Although the result of \Cref{result1} is derived for continuous time systems, it can be extended to discrete time dynamical systems as well. This result is general in the sense that it is applicable irrespective of the method which is used to obtain the approximate input-output simulation. We can also use it to compute upper bounds on the passivity indices of transition system \(T_1\). 

\begin{corollary}
	If \(T_1(\Sigma_1)\preceq^{(\epsilon_u, \epsilon_y)}_{IOS} T_2(\Sigma_2)\) and \(T_2(\Sigma_2)\) is passive with passivity indices \(\rho_2\), \(\nu_2\) then \(T_1(\Sigma_1)\) is also passive with passivity indices \(\rho_1\), \(\nu_1\) which satisfy the following
	\begin{gather}
	\begin{aligned}
	\rho_1(1+\zeta_1\rho_1) \leq \rho_2 - \zeta_3 \\
	\nu_1(1+\zeta_4\nu_1) \leq \nu_2 -\zeta_2,
	\end{aligned}
	\end{gather}
	where \(\zeta_1\), \(\zeta_2\), \(\zeta_3\), \(\zeta_4\) \(\in \mathbb{R}^+\) are arbitrary non-zero constants such that \(\rho_2 - \zeta_3 >0\) and \(\nu_2 -\zeta_2>0\), and \(\rho_2, \nu_2\) are the passivity indices for \(T_2(\Sigma_1)\).
\end{corollary}
\begin{proof}
	Use \Cref{passivity} in \Cref{result1} to obtain this result.
\end{proof}


The result in \Cref{result1} states that if condition (\ref{cond1}) is met, then the QSR dissipativity of an approximate input-output simulation of transition system \(T_1(\Sigma_1)\) implies the QSR dissipativity of \(T_1(\Sigma_1)\) itself. Note that the reverse of this result is not true in general. This can be seen from the definition of approximate input-output simulation. For every transition in \(T_1(\Sigma_1)\) there exists a corresponding approximate transition in \(T_2(\Sigma_2)\). However, \(T_2(\Sigma_2)\) can be a larger system in the sense that there might be some transitions in \(T_2(\Sigma_2)\) for which there is no corresponding transition in \(T_1(\Sigma_1)\). Therefore, from \Cref{result1}, QSR dissipativity of \(T_2(\Sigma_2)\) implies the QSR dissipativity of \(T_1(\Sigma_1)\) and not the other way around in general. In the next part of this section, we consider this reverse problem of determining the \(Q\), $S$ and $R$ dissipativity of an approximate input-output simulation from the QSR dissipativity of the original incrementally forward complete continuous system under a particular abstraction technique.

\subsection{Dissipativity of abstraction of the QSR dissipative system}
This section considers the problem of determining dissipativity matrices of the abstraction of a dissipative system. Zamani et al \cite{zamani2012} showed that the approximate simulation of incrementally forward complete systems can be computed using time and space quantization. We make a slight  modification to this procedure by introducing an extra design parameter to obtain finite abstractions which are approximately input-output similar to the original system. We then quantify the change  in \(QSR\) dissipativity of the system model under such abstraction. Since in this work the abstraction is obtained with respect to measured output of the system, it makes sense to consider two different cases, (i) when the measured output of system is same as system states and, (ii) when measured output is same as the system output. 

\subsubsection{System state as measured output}

In this section we discuss the dissipativity properties of approximate input-output simulation of sampled data systems. The particular class of systems we address here are the ones where measured and actual output of the system are different. For this purpose, we use states as the measured output of $T_\tau(\Sigma)$. However, for dissipativity analysis, we use an alternate output corresponding to $y=h(x,u)$.

There have been several approaches to obtain approximate simulation of systems. Most of them concentrate on a restrictive class of systems. \cite{zamani2012} introduced  a new procedure for construction of abstractions for any non-linear sampled data system which are incrementally forward complete. The approach discussed in \cite{zamani2012} provides sufficient conditions in terms of appropriate sampling time and quantization parameters to obtain countable transition systems guaranteeing approximate (alternating) simulation. We use this technique to construct approximate input-output (alternating) simulation for sampled data systems \(T_\tau(\Sigma)\) and analyze the dissipativity properties of thus obtained abstract system.

\begin{proposition}\label{main_corr1}
	Consider a control system $\Sigma$ in (\ref{sigma}) whose states are the measured output. Given any desired precision parameters $\epsilon_y>0$, $\epsilon_u >0$, if $\Sigma$ satisfies \Cref{incrementally forward complete} then for any $\{\tau,\theta_1,\theta_2,\eta,\mu\} >0$ satisfying
$\eta/2\leq \epsilon_y\leq \theta_1$ and $ \mu/2\leq \epsilon_u\leq \theta_2$, we have:
	\begin{equation}
	T_{\tau,\mu,\eta}(\Sigma)\preceq^{(\epsilon_u, \epsilon_y)}_{IOAS} T_{\tau}(\Sigma)\preceq^{(\epsilon_u, \epsilon_y)}_{IOS} T_{\tau,\mu,\eta}(\Sigma) 
	\end{equation}
	where $T_\tau(\Sigma)$ and $T_{\tau,\mu,\eta}(\Sigma)$ are defined in \Cref{def:discrete-time sys} and \Cref{sampled-quantized system different output} respectively.
\end{proposition}
\begin{proof}
	Proof follows directly from Theorem 4.1 in \cite{zamani2012}. \highlighttext{An outline of the proof can be found in Appendix.}
\end{proof}

 We next analyze the \(QSR\) dissipativity of approximately input-output similar system \(T_{\tau,\mu,\eta}(\Sigma)\) and consider the case where measured output is same as system states.

\begin{theorem} \label{result2}
	Consider a dynamical system $\Sigma$ in (\ref{sigma}) whose states are the measured output.  Suppose $\Sigma$ satisfies Assumptions \ref{incrementally forward complete} and \ref{output rate gain bound} and it is $QSR$ dissipative with respect to the output function $y=h(x,u)$ and \highlighttext{a storage function $V(\cdot): V(x_1)-V(x_2)\leq L \|x_1-x_2\|$}. Let $T_{\tau}(\Sigma)$ be the transition system corresponding to $\Sigma$ with a sampling time $\tau$. If  the input and state quantization parameters $\mu$ and $\eta$ are chosen such that $T_{\tau, \mu, \eta}(\Sigma)$ in \cref{sampled-quantized system different output} is $(\epsilon_u,\epsilon_y)$ - approximately input-output similar to $T_{\tau}(\Sigma)$, then  $T_{\tau, \mu, \eta}(\Sigma)$ is quasi $QSR$ dissipative with matrices \(Q_{\tau, \mu, \eta}\), \(S_{\tau, \mu, \eta}\), \(R_{\tau, \mu, \eta}\) satisfying,
	\begin{gather}
	\label{thm2eq1}
	\begin{aligned}
	Q_{\tau, \mu, \eta} & \geq  Q + \tau\|Q\|_2(\tau\gamma+1)\mathbf{I} \\
	S_{\tau, \mu, \eta} & =  S \\
	R_{\tau, \mu, \eta} & \geq  R + \tau\gamma \|S\|_2\mathbf{I}+ \tau \gamma \|Q\|_2(\tau^2\gamma+\tau+\gamma)\mathbf{I}
	\end{aligned}
	\end{gather}
	
	where \(Q, S, \text{ and } R\) are the dissipativity matrices for \(\Sigma\).
\end{theorem}
\begin{proof}
	See Appendix.
\end{proof}

\subsubsection{System output as measured output}
In the last section, we provided results for the dissipativity properties of approximate input-output simulation for the class of systems where measured output is same as the system states. We now extend these results to the systems where measured output is same as the actual system output. To do this, we make an additional assumption on the system output to be Lipschitz continuous which means that output can not change abruptly. The difference from previous section is that system output \(y=h(x,u)\) is also sampled and quantized here. This can be useful for design of systems with output feedback. 

\begin{proposition}\label{main_corr2}
	Consider a control system $\Sigma$ in (\ref{sigma}) whose measured output is the same as system output. Given any desired precision parameters $\epsilon_y>0$, $\epsilon_u >0$, if $\Sigma$ satisfies \Cref{incrementally forward complete} and the output function is Lipschitz continuous, i.e., \(\|h(x_1,u_1)-h(x_2,u_2)\|\leq K_1 \|x_1-x_2\|+K_2\|u_1-u_2\|\), then for any $\{\tau, \theta_1,\theta_2,\eta,\mu\} >0$ satisfying \highlighttext{$K_1\eta/2+(K_2+1)\mu/2 \leq \epsilon_y$}, $\eta/2\leq \theta_1$ and $ \mu/2\leq \epsilon_u\leq \theta_2$, we have:
	\begin{equation}
	T_{\tau,\mu,\eta}(\Sigma)\preceq^{(\epsilon_u, \epsilon_y)}_{IOAS} T_{\tau}(\Sigma)\preceq^{(\epsilon_u, \epsilon_y)}_{IOS} T_{\tau,\mu,\eta}(\Sigma)
	\end{equation}
    where $T_\tau(\Sigma)$ and $T_{\tau,\mu,\eta}(\Sigma)$ are defined in \Cref{def:discrete-time sys} and \Cref{sampled-quantized system same output} respectively.
\end{proposition}
\begin{proof}
	Proof follows directly from Theorem 4.1 in \cite{zamani2012}. \highlighttext{An outline of the proof can be found in Appendix.} 
\end{proof}

\begin{theorem}\label{result3}
	Consider a dynamical system $\Sigma$ in (\ref{sigma}) whose measured output is the same as system output.  Suppose $\Sigma$ satisfies Assumptions \ref{incrementally forward complete} and \ref{output rate gain bound} and it is $QSR$ dissipative system with respect to the output function $y=h(x,u)$ and \highlighttext{a storage function $V(\cdot): V(x_1)-V(x_2)\leq L \|x_1-x_2\|$}.
	Let $T_{\tau}(\Sigma)$ be the transition system corresponding to $\Sigma$ with a sampling time $\tau$. If  the input and state quantization parameters $\mu$ and $\eta$ are chosen as per \Cref{main_corr2} so that $T_{\tau, \mu, \eta}(\Sigma)$ in \Cref{sampled-quantized system same output} is $(\epsilon_u,\epsilon_y)$ - approximately input-output similar to $T_{\tau}(\Sigma)$, then  $T_{\tau, \mu, \eta}(\Sigma)$ is quasi $QSR$ dissipative with matrices \(Q_{\tau, \mu, \eta}\), \(S_{\tau, \mu, \eta}\), \(R_{\tau, \mu, \eta}\) satisfying,
	\begin{gather}
	\label{thm3eq1}
	\begin{aligned}
    {Q}_{\tau, \mu, \eta} & \geq  Q+\tau\|Q\|_2(\tau\gamma+1)\mathbf{I}  + (\|Q+\tau\|Q\|_2(\tau\gamma+1)\mathbf{I}\|^2_2)\mathbf{I} \\
	{S}_{\tau, \mu, \eta} & =  S \\
	{R}_{\tau, \mu, \eta} & \geq  R+ \tau\gamma\|Q\|_2(\tau^2\gamma+\tau+\gamma)\mathbf{I} +\|S\|^2_2\mathbf{I} + (\gamma\sqrt{m\tau}\mu +\gamma^2\tau)\mathbf{I}
	\end{aligned}
	\end{gather}
	where \(Q, S, \text{ and } R\) are the dissipativity matrices for \(\Sigma\).
\end{theorem}
\begin{proof}
	See Appendix.
\end{proof}


\begin{remark}\normalfont
	\begin{enumerate}
		\item The inequalities in (\ref{cond1}) restrict only the spectral radius of the dissipation matrices. This gives some flexibility in choosing the actual structure of the \(Q\) and \(R\) matrices.
		\item Since passivity is a special case of \(QSR\) dissipativity, results guaranteeing the passivity levels for abstract system can be derived from Theorems \ref{result2} and \ref{result3} in a straight forward manner.
		\item It should be noted that the notion of approximate input-output simulation here, under mild conditions, is composition preserving \cite{tazaki-hscc-2008} in that the interconnection of two abstractions is the same as the abstraction of interconnection of two systems. Also, passivity is preserved over feedback and parallel interconnections \cite{willems}. Therefore, the composition of passivity of abstractions is the same as the passivity of abstraction of composition. This is an interesting result that can be used to reason about the passivity properties of abstractions of large scale systems.
		\item Symbolic models obtained using abstractions have been used for design and analysis of control systems \cite{girard3}. However, issues like stability are in general difficult to address using symbolic models. Theorems \ref{result2} and \ref{result3} can be used to infer dissipativity properties of symbolic models obtained using approximate input-output simulation based abstractions. Once the system design is complete, the dissipativity properties of symbolic model can be translated back to that of the original system using the discrete version of results in Theorem \ref{result1}. Since, QSR dissipativity implies system stability under certain conditions, apart from offering compositionality, these results on QSR dissipativity can be used to guarantee stability of system abstractions as well.
	\end{enumerate}
\end{remark}

\section{Examples}
\label{example}

In this section, we present two illustrative examples. We describe two simple linear time invariant (LTI) systems and discuss how the ideas in this paper can be used to obtain a symbolic model and carry out the dissipativity analysis. For both these examples, we consider input feed-forward output feedback passivity which is a special case of \(QSR\) - dissipativity.

\noindent\textbf{Example 1.} Consider an LTI system $ \Sigma:\dot{x}=-x+u$, with the measured output described by an identity map, i.e., $y_m=x$. For dissipativity analysis, we consider another output function $y = h(x,u) = Cx+Du = x+u$. It can be verified that this system  is input feed-forward output feedback passive with respect to an output function $y=x+u$ and \(V(x) = \frac{1}{2}x^TPx=\frac{1}{2}x^T(0.5154)x\). The passivity index are $(0.25,0.5)$.

Now we construct an approximately input-output similar symbolic model for $\Sigma$. Based on the discussion on incremental forward completeness of linear systems in Section V of \cite{zamani2012}, it is readily seen that $\Sigma$ is  incrementally forward complete, thus we can apply \Cref{main_corr1}. We work on the subset $X=[-0.2, 0.2]$ of the state space and subset $U=[-0.1,0.1]$ of the input space. To construct the symbolic model  $T_{\tau,\mu, \eta}(\Sigma)$ of precision $\epsilon_u=0.1,\epsilon_y=1$, choose $\theta_1=1$, $\eta=0.1$, $\theta_2=\epsilon_u=\mu=0.1$ and $\tau=0.2$ so that the conditions in \Cref{main_corr1} are satisfied. Since $\mu=0.1$ and $\tau=0.2$, the control input is piecewise constant signal of duration $\tau$ such that 
$$\{-\mu, 0,\mu \}=\{u_{-1},u_0,u_1\}=\{-0.1,0,0.1\} \in U_q,$$
and the states of the symbolic system are described by 
$$\{-2\eta, -\eta, 0, \eta, 2\eta \}=\{-0.2, -0.1,0,0.1,0.2\} \in X_q.$$

 \begin{figure*}[htp]
 \vspace{-1em}
 	\centering
 	\begin{tikzpicture}[->,>=stealth',shorten >=1pt,auto,node distance=2.5cm,
 	semithick]
 	\tikzstyle{every state}=[fill=gray,draw=none,text=white]
 	
 	\node[state] (A)                    {$0$};
 	\node[state]         (B) [ right of=A] {$\eta$};
 	\node[state]         (C) [ right of=B] {$2\eta$};
 	\node[state]         (D) [left of=A]    {$-\eta$};
 	\node[state]         (E) [left of=D]     {$-2\eta$};
 	
 	\path (E) edge              node {$u_{-1}, u_0, u_1$} (D)
 	edge        [loop above] node {$u_{-1},u_0,u_1$} (E)
 	(D) edge    [loop above] node {$u_{-1},u_0, u_1$} (D)
 	edge              node {$u_0,u_1$} (A)
 	(A) edge  [loop above] node {$u_{-1}, u_0, u_1$} (A)
 	edge   [bend left] node {$u_{-1}$} (D)
 	edge  [bend left] node {$u_1$}(B)
 	(B) edge                 node {$u_{-1}, u_0$}(A)
 	edge   [loop above] node {$u_{-1}, u_0,u_1$} (B)
 	(C) edge                 node {$u_{-1},u_0,u_1$} (B)	
 	edge  [loop above] node {$u_{-1},u_0,u_1$} (C);
 	
 	\end{tikzpicture}
 	\caption{Symbolic model for $\Sigma$.}
 	\label{fig:state_machine}
    \vspace{-0.5em}
 \end{figure*}
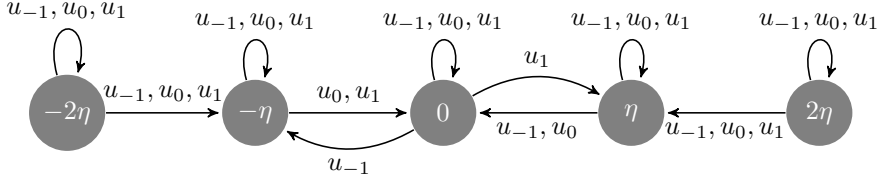

The transitions between states upon the action of a control input can be calculated using the differential equation describing $\Sigma$.  As seen in \Cref{fig:state_machine}, the symbolic model $T_{\tau,\mu, \eta}(\Sigma)$ is non-deterministic. 

The effect of  symbolic abstraction on the passivity properties of $\Sigma$ using \Cref{result2}. It can be verified that the output $y=x+u$ satisfies \Cref{output rate gain bound} for $\gamma=1$. Hence, using \Cref{result2} we can state that $T_{\tau, \eta,\mu}(\Sigma)$ is  $\left(\rho_{\tau, \mu, \eta}, \nu_{\tau, \mu, \eta}\right)$- input feedforward output feedback quasi passive where
\begin{align}
\rho_{\tau, \mu, \eta} & =\rho - \tau\|\rho\|_2(\tau\gamma+1) =0.19 \\
\nu_{\tau, \mu, \eta} & = \nu - \tau\gamma \|0.5\|_2- \tau \gamma \|\rho\|_2(\tau^2\gamma+\tau+\gamma) = 0.338
\end{align}

For the symbolic transition system, \Cref{result2}  can be alternatively verified by checking if 
$${\bm \ell}^T o-\rho_{\tau, \mu, \eta} \mathbf{o}^T\mathbf{o}- \nu_{\tau, \mu, \eta}{\bm \ell}^T{\bm \ell} \geq \hat{V}(\mathbf{p})-\hat{V}(\mathbf{q})-\beta$$
 is satisfied for all  transitions  $\mathbf{q} \xrightarrow [ \quad \tau  \quad]{{\bm \ell}} \mathbf{p}$ where $\mathbf{p} \in \textbf{Post}_{{\bm \ell}}(\mathbf{q})$, $\mathbf{o}=C\mathbf{q}+D{\bm \ell}$, $\hat{V}(\mathbf{q})=\frac{1}{2\tau}\mathbf{q}^TP\mathbf{q}$ and $\beta = \frac{L\eta}{2\tau}$, i.e., 
 $\mathbf{q}^TF\mathbf{q}+{\bm \ell}^TG\mathbf{q}+\mathbf{q}^TG^T{\bm \ell}+{\bm \ell}^TH{\bm \ell}+\tau\beta-\frac{1}{2}\mathbf{p}^TP\mathbf{p}\geq 0 $
where
$ F = \frac{1}{2}P-\rho_{\tau, \mu, \eta}\tau C^TC,\ G=\frac{\tau}{2} C-\rho_{\tau, \mu, \eta}\tau D^TC,\
H = \frac{\tau}{2}(D+D^T)-\rho_{\tau, \mu, \eta} \tau D^TD-\nu_{\tau, \mu, \eta} \tau I.$
Candidate values for $\hat{V}(\cdot)$ and $\beta$ are obtained from the proof of \Cref{result2}.

 \noindent For the symbolic system, we assume that there are $M$  quantized inputs denoted by $\{{\bm \ell}_1,{\bm \ell}_2, \dots, {\bm \ell}_M\}$ and there are $N$  quantized states denoted by $\{\mathbf{q}_1,\mathbf{q}_2,\dots,\mathbf{q}_N\}$. 
All the transitions in the symbolic system can be represented by $\mathbf{q}_i \xrightarrow [ \quad \tau \quad]{{\bm \ell}_j} \mathbf{p}^{j}_{i}$ for $i=1, \dots N$ and $j=1, \dots, M,$
where $\mathbf{p}_i^j$ represents the next state after time $\tau$ with an initial state $\mathbf{q}_i$, under the action input ${\bm \ell}_j$.  Hence, passivity verification would entail  verification of the inequality
\begin{eqnarray}\label{passive_scalar_ineq}
\mathbf{q}_i^TF\mathbf{q}_i+{\bm \ell}_j^TG\mathbf{q}_i+\mathbf{q}_i^TG^T{\bm \ell}_j+{\bm \ell}_j^TH{\bm \ell}_j+\tau\beta-\frac{1}{2}(\mathbf{p}_i^j)^TP(\mathbf{p}_i^j)\geq 0
\end{eqnarray}
\noindent for $i=1, \dots N$ and $j=1, \dots, M.$  In order to verify the above inequality for all transitions in a systematic fashion, we
let $\bar{\mathbf{q}}=\begin{bmatrix}\mathbf{q}_1, \cdots, \mathbf{q}_N\end{bmatrix}^T$, $\bar{{\bm \ell}}=\begin{bmatrix}{\bm \ell}_1, \cdots,{\bm \ell}_M\end{bmatrix}^T$ and arrange vectors $\bar{\mathbf{p}}^1=\begin{bmatrix}\mathbf{p}_1^1, \cdots,\mathbf{p}_N^1\end{bmatrix}^T, \dots, \bar{\mathbf{p}}^M=\begin{bmatrix}\mathbf{p}_1^M, \cdots,\mathbf{p}_N^M\end{bmatrix}^T $ together as $\bar{\bar{\mathbf{p}}}=\begin{bmatrix}\bar{\mathbf{p}}^1, \cdots, \bar{\mathbf{p}}^M\end{bmatrix}^T$.
Verification of  \eqref{passive_scalar_ineq}  $\text{for }i=1, \dots N\text{ and }j=1, \dots, M$  would require us to verify positivity of $MN$ scalars. All these $MN$ scalars will be arranged along the diagonal of an $MN \times MN$  matrix, and this diagonal matrix would be checked for its positive definiteness. This approach allows us to represent all the inequalities together in a compact fashion. This compact representation will be achieved using the Kronecker product as given by 
\begin{small}
\begin{eqnarray}\label{passive_kron}
PASSIVE&=&I_M\otimes ((I_N\otimes \bar{\mathbf{q}}^T)(I_N\otimes F)(I_N\otimes \bar{\mathbf{q}})) + ((I_N\otimes \bar{{\bm \ell}}^T)(I_N\otimes G)(I_N\otimes \bar{\mathbf{q}}))\otimes I_M \nonumber\\
&+&((I_N\otimes \bar{\mathbf{q}}^T)(I_N\otimes G^T)(I_N\otimes \bar{{\bm \ell}}))\otimes I_M +((I_M\otimes \bar{{\bm \ell}}^T)(I_M\otimes H)(I_M\otimes \bar{{\bm \ell}}))\otimes I_N \nonumber\\&+&I_{MN} \otimes \frac{K\epsilon_y}{\tau}-\bar{\bar{\mathbf{p}}}^T(I_{MN}\otimes P)\bar{\bar{\mathbf{p}}} \geq 0
\end{eqnarray}
\end{small}
\noindent For the nondeterministic  cases where $\bar{\bar{\mathbf{p}}}$ is not unique, we verify \eqref{passive_kron} for all possible values of $\bar{\bar{\mathbf{p}}}$ . Performing this test for our numerical example, we obtain the diagonal elements of the PASSIVE matrix for two possible values of $\bar{\bar{\mathbf{p}}}$ and it can verified that all the diagonal elements are positive, hence confirming the passivity of the symbolic model.

This example demonstrates that the results in this paper can be used to avoid large computations for determining passivity of discrete abstractions of continuous systems. 

\noindent \textbf{Example 2.} In this example we use the results of \Cref{result3} to validate the passivity performance of a plant connected with a controller implemented in software. Consider a linear time invariant system,
\begin{equation}
\label{example_system2}
\dot{x}=Ax+Bu, \qquad
y = Cx+Du
\end{equation}
\begin{align}
A = \begin{bmatrix}
-3.6 & 0.2 & 2.4 & 0 & 0\\
0.2 & -1 & 0 & -0.6 & 0\\
2.4 & 0 & -6 & -4 & 1\\
0 & -0.6 & -4 & -6 & -0.8\\
0 & 0 & 1 & 0.8 & -2
\end{bmatrix},\quad
B = \begin{bmatrix}
0.1 \\ 0.4 \\ 0.1 \\ 0.5 \\ 0.1
\end{bmatrix},\quad
C  = B^T,\quad
D = \begin{bmatrix}
0.2
\end{bmatrix}. \nonumber
\end{align}
It can be verified that the system in (\ref{example_system2}) is passive with passivity indices $(0.15,0.7)$. The $\mathcal{L}_2$ gain that bounds the rate of change of output $y$ is $\gamma=CB=0.44$. Suppose this system is connected in feedback to a passive LTI controller
\begin{equation}
\dot{z}=A_cz+B_cw, \qquad
v = C_cz+D_cw
\end{equation}
\begin{align}
A_c=\begin{bmatrix}
-2 & -1\\
-3 & -5
\end{bmatrix}, \quad
B_c = \begin{bmatrix}
0.1 \\
0.2
\end{bmatrix}, \quad
C_c  = \begin{bmatrix}
1 & 1
\end{bmatrix}, \quad D_c = \begin{bmatrix}
1
\end{bmatrix} \nonumber
\end{align}
implemented in software. For a sampling time of 0.2 sec and the state, input and output quantization value of 0.1, the passivity indices of this controller are $(\rho_c,\nu_c)=(0.0420,0.8115)$. For the ease of analysis, we focus on the subset $[-0.2,0.2]$, $[-0.1,0.1]$ and $[-0.1,0.1]$ of the state, input and output space respectively. The controller symbolic model can be obtained 
by considering piecewise continuous control inputs $\{u_{-1},u_0,u_1\}=\{-0.1,0,0.1\}$ and the symbolic states and outputs described by $\{-x_2, -x_1, x_0, x_1, x_2 \}=\{-0.2, -0.1,0,0.1,0.2\}$ and $\{y_{-1},y_0,y_1\}=\{-0.1,0,0.1\}$ respectively.

As shown in \Cref{fig:sys_soft}, the system (\ref{example_system2}) is interacting with the software controller through continuous to discrete and discrete to continuous conversion units. This system can be analyzed by considering a discrete abstraction of the continuous plant. Using \cref{result3} for passivity,
\begin{align}
\rho_{\tau, \mu, \eta} = \rho - \tau\rho(\tau\gamma+1)-|1-\rho+\tau\rho(\tau\gamma+1)| = -0.7653\nonumber \\
\nu_{\tau, \mu, \eta} = \nu -1/2 - \tau\gamma\rho(\tau^2\gamma+\tau+\gamma) - (\gamma\sqrt{m\tau}\mu+\gamma^2\tau)= 0.1329 \nonumber
\end{align}

Clearly, this discrete abstraction of the plant is no longer output feedback passive ($\rho_{\tau, \mu, \eta}<0$). However, following Theorem 7 in \cite{zhu}, since \(\nu_{\tau, \mu, \eta} > 0\), \(\rho_c >0\) and \(\rho_{\tau, \mu, \eta}+\nu_c > 0\), the closed loop system is passive with output passivity index of 0.0462.
\section{Dissipativity of approximate feedback composition of systems}
\label{composition}
In this section we discuss the dissipativity property of the approximate feedback composition of two transition systems as described in \cite{tabuada2009}. We show that once two transition systems are approximately feedback composable, then QSR dissipativity of one of those transition systems implies QSR dissipativity of the entire composition. 

Cyber physical systems can be constructed by interconnecting  several individual subsystems and this process for transition systems can be described using composition operations. It was shown in \cite{tabuada2009} that approximate feedback composition of two transition systems can also be used to construct controllers for requirements such as safety and reachability.  Approximate feedback composition of two transition systems is possible for state feedback if there exists an approximate alternating simulation relation between the two systems that may be a plant and a controller. The idea of supervisory control in \cite{tabuada2009} is that the controller restricts the behavior of the plant by forcing it to simulate the controller.
The controller selects an allowable input label, the plant
makes any transition having that input label, and the controller makes a transition to maintain alternating simulation relation. This set up works if there is room for the controller to select its input to properly navigate the plant behavior while maintaining the approximate alternating simulation relation. In the original concept of approximate alternating simulation in \cite{tabuada2009}, the input sets of two systems which are approximately alternating similar were different. 
This freedom of nonidentical inputs of two transition systems along with the other conditions of approximate alternating simulation defined in \cite{tabuada2009} are captured by \((\epsilon_u,\epsilon_y)\) approximately input-output alternating simulation in definition 4, with \(\epsilon_u \neq 0\). Therefore, we can modify the definition of approximately feedback composable systems in \cite{tabuada2009} from using approximate alternating similar systems to \((\epsilon_u,\epsilon_y)\) approximate input-output alternating similar systems. This will be clear in the following definition.

\begin{definition} A transition system $T_2$ is said to be $(\epsilon_u,\epsilon_y)$-approximate feedback composable with system $T_1$ if there exists an $(\epsilon_u,\epsilon_y)$-approximate input-output alternating simulation relation $\mathcal{R}$ from $T_2$ to $T_1$, that is, \(T_2\preceq^{(\epsilon_u, \epsilon_y)}_{IOAS} T_1\).
\end{definition}


Let $ T_{i}:=(X_i,U_i, \xrightarrow [\quad \tau \quad]{},Y_i, \mathcal{H}_i)$, \(i = \{1,2\}\) be two transition systems with a common time period $\tau$ and common input   and output sets equipped with euclidean norm as the metric.  Let $\mathcal{R}$ be a $(\epsilon_u,\epsilon_y)$ - approximate input-output alternating simulation relation from $T_2$ to $T_1$. Let us define a feedback relation $\mathcal{F}\subset X_1\times X_2 \times U_1 \times U_2$ defined by all the quadruples $(x_1,x_2,u_1,u_2)\in X_1 \times X_2 \times U_1 \times U_2$ for which $(x_1,x_2)\in \mathcal{R}$ and conditions 1 and 2 in \Cref{def:IOAS} are met.

The feedback composition of $T_2$ and $T_1$ with interconnection relation $\mathcal{F}$, denoted
by $T_2 \times_{\mathcal{F}}^{(\epsilon_u,\epsilon_y)}T_1$, is the transition  system   $(X_{12}, U_{12}, \xrightarrow [ \quad \tau \quad]{} , Y_{12}, \mathcal{H}_{12})$ consisting of
\begin{itemize}
	\item $X_{12}=\{({x}_1, {x}_2) \in (X_1 \times X_2)\,\, \vline \,\, \mathbf{d_Y}(\mathcal{H}_1({x}_1,u_1), \mathcal{H}_2({x}_2,u_2))\leq \epsilon_y\}$, where $\mathbf{d_U}(u_1,u_2)\leq \epsilon_u$;
	\item $U_{12}=\{(u_1,u_2)|\mathbf{d_U}(u_1,u_2)\leq \epsilon_u, u_1\in U_1, u_2 \in U_2\}$;
	\item $({x}_1,{x}_2)  \xrightarrow [ \quad \tau \quad]{({ u_1},{ u_2})} ({x}_1',{x}_2')$ if the following three conditions hold:
	\begin{enumerate}
		\item ${x}_1 \xrightarrow [ \quad \tau \quad]{u_1} {x}_1'$ in $T_1$;
		\item ${x}_2 \xrightarrow [ \quad \tau \quad]{u_2} {x}_2'$ in $T_2$;
		\item $({x}_1,{x}_2,{ u_1},{u_2}) \in \mathcal{F} $; 
	\end{enumerate}
	\item $Y_{12}=Y_1= Y_2$;
	\item $\mathcal{H}_{12}(x_1, x_2,u_1,u_2) = \frac{1}{2}(\mathcal{H}_1(x_1,u_1)+ \mathcal{H}_2(x_2,u_2))$.
\end{itemize} 

This symmetrical choice of output allows $T_2 \times_{\mathcal{F}}^{(\epsilon_u,\epsilon_y)}T_1$ to be commutative. However, we can also choose an output for the composition as $\mathcal{H}_{12}(x_1, x_2,u_1,u_2) = \mathcal{H}_1(x_1,u_1)$ or $\mathcal{H}_{12}(x_1, x_2,u_1,u_2) = \mathcal{H}_2(x_2,u_2)$.\newline 

Before analyzing the dissipativity of the feedback composition, we present the following result from \cite{tabuada2009}. Even though  the results in \cite{tabuada2009} were derived for approximate simulation relationships   they also hold true for  approximate input-output simulation relationships.

\begin{proposition} \label{prop1} Consider approximate feedback composition of two \((\epsilon_u,\epsilon_y)\) approximately input-output similar transition systems $T_1$ and $T_2$, where \(T_2\preceq^{(\epsilon_u, \epsilon_y)}_{IOAS} T_1\). If we define the output of the composition as 
	\begin{enumerate}
		\item $ h_{12}({x}_1, {x}_2,u_1,u_2) = \frac{1}{2}(h_1({x}_1,u_1)+ h_2({x}_2),u_2)$ then,
		\begin{enumerate}
			\item $T_2 \times_{\mathcal{F}}^{(\epsilon_u,\epsilon_y)}T_1\preceq_{IOS}^{(\epsilon_u,\epsilon_y/2)} T_2 $,
			\item $T_2 \times_{\mathcal{F}}^{(\epsilon_u,\epsilon_y)}T_1\preceq_{IOS}^{(\epsilon_u,\epsilon_y/2)} T_1 $
		\end{enumerate}
		\item $h_{12}({x}_1, {x}_2,u_1,u_2)=h_1({x}_1,u_1)$, then  $$T_2 \times_{\mathcal{F}}^{(\epsilon_u,\epsilon_y)}T_1\preceq_{IOS}^{(\epsilon_u,\epsilon_y)} T_2 $$
		\item $h_{12}({x}_1, {x}_2,u_1,u_2)=h_2({x}_2,u_2)$, then  $$T_2 \times_{\mathcal{F}}^{(\epsilon_u,\epsilon_y)}T_1\preceq_{IOS}^{(\epsilon_u,\epsilon_y)} T_1. $$
	\end{enumerate}
\end{proposition}

\begin{proof}  This is a direct consequence of Proposition 11.8 in \cite{tabuada2009}. 
\end{proof}

\noindent Based on \Cref{prop1} and \Cref{result1} we obtain the following result.

\begin{theorem}\label{passive_composition}
	\hspace{-1mm}Let $ T_{1}\hspace{-1mm}:=\hspace{-1mm}(X_1,U_1, \xrightarrow [\quad \tau \quad]{}, Y_1, \mathcal{H}_1)$ and  $ T_2\hspace{-1mm}:=\hspace{-1mm}(X_2,U_2, \xrightarrow [\quad \tau \quad]{}, Y_2, \mathcal{H}_2 )$ be two transition systems with a common time period $\tau$ and common input and output sets equipped with euclidean norm as the metric. Let $T_2$ be $(\epsilon_u,\epsilon_y)$ - approximately input-output alternatingly similar to $T_1$, \(T_2\preceq^{(\epsilon_u, \epsilon_y)}_{IOAS} T_1\). If $T_1$ is $QSR$ - dissipative with respect to an output function $h(x_1, u_1)$ where  $x_1\in X_1$ and $ {u}_1 \in U_1$, then $T_2 \times_{\mathcal{F}}^{(\epsilon_u,\epsilon_y)}T_1$  is also $QSR$ dissipative w.r.t. $\frac{1}{2}\left(h_1(x_1,u_1)+ h_2(x_2,u_2)\right)$ where $(x_1,x_2,u_1,u_2) \in \mathcal{F}$, and the dissipation matrices \(Q_{12}\), \(S_{12}\), \(R_{12}\) of the composed system satisfy
    \begin{gather}\label{cond1thm4}
    \begin{aligned}
	\underline{\lambda}(Q_{12}-Q_1)-\zeta_1\|Q_{12}\|^2_2 -\zeta_3 &\geq 0\\
	\underline{\lambda}(R_{12}-R_1)- \zeta_2\|S_{12}\|^2_2 -\zeta_4 \|R_{12}\|^2_2& \geq 0 \\
	S_{12} &= S_1
	\end{aligned}
    \end{gather}	
	where \(\zeta_1\), \(\zeta_2\), \(\zeta_3\), \(\zeta_4\) \(\in \mathbb{R}^+\) are arbitrary non-zero constants, \(Q_1, S_1, R_1\) are the dissipativity matrices for \(T_1\). \(\underline{\lambda}(\cdot)\) and \(\overline{\lambda}(\cdot)\) represent the smallest and largest eigen values of the concerned matrix.
	
	Also, $T_2 \times_{\mathcal{F}}^{(\epsilon_u,\epsilon_y)}T_1$  is $QSR$-dissipative w.r.t. $h_2(x_2, {u_2})$ where $x_2 \in X_2$ such that $(x_1,x_2,u_1,u_2) \in X_{12} \times U_{12}$ with the dissipation matrices \(Q_{12}\), \(S_{12}\), \(R_{12}\) as in (\ref{cond1thm4}).
	
\end{theorem}
\begin{proof}
	See Appendix.
\end{proof}

This result states that once two transition systems are approximately feedback composable, then QSR dissipativity of one of those transition systems implies QSR dissipativity of the composed transition system. Since passivity is a special case of QSR dissipativity, \Cref{passive_composition} can be applied to design discrete supervisory controllers for plants, while preserving the passive nature of the interconnection. 

The continuous-time system should be preceded by a sample and hold element to convert the common quantized input symbol into a piecewise constant input. Under this framework, the interconnected system is passive w.r.t. outputs 
\begin{itemize}
	\item[(i)] $h_2(x_2,u_2)$, where $x_2$ is a discrete state of the controller and 
	\item[(ii)] $\frac{1}{2}\left(h_1(x_1,u_1)+h_2(x_2,u_2)\right)$ where $x_1$ is the discrete plant state and $x_2$ is the discrete controller state.
\end{itemize}

\noindent Once the interconnection is passive, we can guarantee the stability of this interconnection \cite{willems}. Further discussion on supervisory control using passivity is a subject of our future work.

\section{Conclusion}
\label{conclusion}
In this paper we characterized the dissipativity properties of a system and its abstracted model. First we presented results to compute the \(Q, \text{\space}S\), and $R$ dissipativity matrices of transition system from those of its approximate input-output similar abstraction. We also provided conditions determining the dissipativity matrices of an abstract system from those of the corresponding incrementally forward complete continuous system. Abstraction was obtained by an approximate input-output simulation of the continuous system. Further, we considered the approximate feedback composition of a continuous QSR dissipative system with a finite state transition system that may be a symbolic controller. We showed that if one of the components in this interconnection is QSR dissipative,  then the approximate feedback composition is also QSR dissipative.

We can also use the results presented here to design symbolic controllers to passivate and hence stabilize a dynamical system. Although a lot of work has been done to design continuous passivating controllers for a variety of dynamical systems, it is challenging to design such continuous controllers when other non-traditional control constraints, for example expressed in terms of temporal logic, need to be met along with passivity specifications. In such cases, the importance of supervisory or discrete controllers is more apparent.  As a future work, we will concentrate on designing controllers that passivate the system along with ensuring other system specifications such as safety and reachability.



\appendix
\section*{Appendix}
\subsection*{Proof of \Cref{result1}}
\begin{proof}
Suppose at any time $t$, the input, output and state of transition system $T_i(\Sigma_i)$ be $u_i(t), \ y_i(t)$ and $x_i(t)$ respectively. Define the supply rate function $\omega_2(u_2(t),y_2(t))=y_2^T(t)Q_2y_2(t) + 2 y_2^T(t)S_2u_2(t) + u_2^T(t)R_2u_2(t)$ and storage function $V(\cdot)$. Since $T_2$ is $QSR$ dissipative with dissipativity matrices $Q_2,\ S_2$ and $R_2$, $r_2 = \int_{t_0}^{t_1}\omega_2(u_2(t),y_2(t))dt\geq V(x_2(t_1))-V(x_2(t_0)).$
Consider another quadratic supply rate function $\omega_1(u_1(t),y_1(t)) =y_1^T(t)Q_1y_1(t) + 2 y_1^T(t)S_1u_1(t) + u_1^T(t)R_1u_1(t).$ Let this be the quadratic supply rate for transition system  \(T_1\). Define $r_1 = \int_{t_0}^{t_1} \omega_1(u_1(t),y_1(t)) dt$.

$ T_2$ is $(\epsilon_u,\epsilon_y)$ - approximately input-output similar to $T_1$ with an approximate input-output simulation relation $\mathcal{R}$. Hence we can always find a transition $x_2 \xrightarrow [\quad t_1 \quad]{u_{\textbf{2}}} x_2'$ in $T_2$ for every transition $x_1 \xrightarrow [\quad t_1 \quad]{u_1}x'_1$ in $T_1$  such that \highlighttext{$\mathbf{d_U}(u_1,u_2)=\|u_1-u_2\|=\|-\Delta u \|\leq \epsilon_u$ and \(\mathbf{d_Y}(h_1(x_1,u_1),h_2(x_2,u_2))=\|y_1-y_2\|=\|-\Delta y \|\leq \epsilon_y\) with \((x_1,x_2)\in\mathcal{R}\) and \((x_1',x_2')\in\mathcal{R}\). For ease of representation, we drop the time index from input and output signals.}
\small
\begin{align}
r_1 = \hspace{-2mm}\int_{t_0}^{t_1} ((y_2-\Delta y)^T Q_1 (y_2 - \Delta y)+ 2 (y_2 - \Delta y)^T S_1 (u_2-\Delta u)+ (u_2-\Delta u)^T R_1 (u_2-\Delta u) )dt\nonumber\end{align} \vspace{-3mm}
\begin{align}\label{dif_eq}
&r_1-r_2 = \int_{t_0}^{t_1} (y_2^T (Q_1-Q_2)y_2 + 2 y_2^T (S_1-S_2)u_2 +  u_2^T (R_1-R_2) u_2 - 2 y_2^T Q_1 \Delta y \nonumber \\
 &\qquad \qquad- 2 \Delta y^T \hspace{-1mm} S_1 u_2- 2  y_2^T\hspace{-1mm} S_1 \Delta u-2 u_2^T\hspace{-1mm} R_1 \Delta u + 2 \Delta y^T\hspace{-1mm} S_1 \Delta u+ \Delta u^T \hspace{-1mm}R_1 \Delta u + \Delta y^T\hspace{-1mm} Q_1 \Delta y)dt \nonumber\\
\end{align} 
\vspace{-10mm}
 \begin{gather}\label{dif_eq1}
\begin{aligned}
2\langle y_2, Q_1 \Delta y\rangle_{t_1} & \leq \frac{\epsilon_y}{\zeta_1}+\zeta_1\|Q_1\|^2_2\langle y_2, y_2\rangle_{t_1}, \quad
2\langle \Delta y, S_1 u_2\rangle_{t_1} &\leq \frac{\epsilon_y}{\zeta_2}+\zeta_2\|S_1\|^2_2\langle u_2, u_2\rangle_{t_1}\\
2 \langle  y_2, S_1 \Delta u\rangle_{t_1} &\leq \frac{\epsilon_u}{\zeta_3}\|S_1\|^2_2 + \zeta_3 \langle y_2, y_2\rangle_{t_1}, \quad
2\langle u_2, R_1 \Delta u\rangle_2 &\leq \frac{\epsilon_u}{\zeta_4}+\zeta_4\|R_1\|^2_2\langle u_2, u_2\rangle_{t_1}\\
\langle \Delta u, R_1 \Delta u \rangle_{t_1} &\geq \underline{\lambda}(R_1)\langle \Delta u , \Delta u \rangle_{t_1},\qquad \quad
\langle \Delta y, Q_1 \Delta y\rangle_{t_1} &\geq \underline{\lambda}(Q_1)\langle \Delta y , \Delta y \rangle_{t_1}
\end{aligned}
\end{gather}
\normalsize
where \(\underline{\lambda}(\cdot)\) is the smallest eigen value of matrix under consideration and \highlighttext{notation $\langle y, y\rangle_{t_1} = \int_{t_0}^{t_1} y^T(t)y(t)dt$}.

Substituting the set of equations (\ref{dif_eq1}) in (\ref{dif_eq}), we can show that if (\ref{cond1}) is satisfied then 
then \(r_1> \beta_1\), where \(\beta_1 = r_2 - \frac{\epsilon_y}{\zeta_1} - \frac{\epsilon_y}{\zeta_2}-\frac{\epsilon_u}{\zeta_3}\|S_1\|^2_2 - \frac{\epsilon_u}{\zeta_4} - max\{0, \overline{\lambda}(-Q_1)\}\epsilon_y - max \{0, \overline{\lambda}(-R_1)\}\epsilon_u - max\{0, -2\langle \Delta y, S_1 \Delta u \rangle_{t_1}\} \). Since $T_2$ is $QSR$ dissipative, as per Theorem 3.1.11 of \cite{van der schaft book 2017}, the available storage 
$S_{a_2}(x_2(t_0)) = \sup\limits_{u_2(\cdot),t_1\geq t_0} - r_2(t_1) < \infty,$ $x_2(t_0) = x_2 \ \forall x_2 \in X_2$. In particular, this is true for any $x_2:(x_1,x_2)\in \mathcal{R}$. Clearly, available storage $S_{a_1}(x_1) = \sup\limits_{u_1(\cdot),t_1\geq0} - r_1(t_1) < \infty,\ x_1(t_0) = x_1 \ \forall x_1 \in X_1$. Therefore, from Theorem 3.1.11 of \cite{van der schaft book 2017}, $T_1$ is $QSR$ dissipative.
\end{proof}
\subsection*{Proof of \Cref{main_corr1}}
\begin{proof}
\highlighttext{We first show that  $ T_{\tau}(\Sigma)\preceq^{(\epsilon_u, \epsilon_y)}_{IOS} T_{\tau,\mu,\eta}(\Sigma)$. Let us define a relation $\mathcal{R}\subset X_\tau \times X_q$ such that $(x_\tau,x_q)\in \mathcal{R}$ if $\|x_{\tau}-x_q\| \leq \eta/2$. Note that for any $x_{\tau} \in X_{\tau}$ and $u_{\tau} \in U_{\tau}$, there always exist $x_q \in X_{\textbf{q}}$ and $u_q \in U_{\textbf{q}}$ such that $\|x_{\tau}-x_q\| \leq \eta/2 \leq \epsilon_y$ and $\|u_{\tau}-u_q\| \leq \mu/2 \leq \epsilon_u$. This is possible because of the specific quantization which allows $x_{\tau}$ to be within $\eta/2$ radius of $x_{q}$ and $u_{\tau}$ to be within $\mu/2$ radius of $u_q$. From the definitions of output functions $H_{m_\tau}(x_\tau,u_\tau)=x_\tau$ and $H_{m_q}(x_q,u_q)=x_q$, we have $\|H_{m_\tau}(x_{\tau},u_\tau)-H_{m_{q}}(x_{\textbf{q}},u_\tau)\|=\|x_{\tau}-x_q\| \leq \epsilon_y$, hence condition (i) of \Cref{def:IOS} is satisfied.
\newline Now if we consider the transition $x_{\tau} \xrightarrow [\quad \tau \quad]{u_{\tau}} x_{\tau}'$ in the transition system $T_{\tau}(\Sigma)$ , then the distance between $x_{\tau}'$  and  $\xi(\tau, x_{\textbf{q}},u_{\textbf{q}})$ can estimated based on the incremental forward complete property of $\Sigma$,
\begin{equation*}
\|x_{\tau}'-\xi(\tau, x_{\textbf{q}},u_{\textbf{q}})\|\leq \alpha_1(\eta,\tau)+\alpha_2(\mu, \tau)\leq \alpha_1(\epsilon_y,\tau)+\alpha_2(\epsilon_u, \tau)
\end{equation*}
As mentioned earlier, due to the particular structure of quantization, for any $x_\tau'\in X_\tau$ there always exists $x_q'\in X_q$ such that 
\begin{equation} \label{ineq_eta1}
\|x_{\tau}'-x_{\textbf{q}}'\| \leq \eta/2
\end{equation}. 
From the triangular inequality we have
\begin{align*}
\|\xi(\tau, x_{\textbf{q}},u_{\textbf{q}})-x_{\textbf{q}}'\|&\leq\|\xi(\tau, x_{\textbf{q}},u_{\textbf{q}})-x_{\tau}'\|+ \|x_{\tau}'-x_{\textbf{q}}'\|\\
&\leq\alpha_1(\epsilon_y,\tau)+\alpha_2(\epsilon_u, \tau)+ \eta/2\\
&\leq\alpha_1(\theta_1,\tau)+\alpha_2(\theta_2, \tau)+ \eta/2
\end{align*}
which, by the \Cref{sampled-quantized system different output} of $T_{\tau, \mu, \eta}(\Sigma)$ implies the existence of $x_{\textbf{q}} \xrightarrow [\quad \quad]{u_{\textbf{q}}} x_{\textbf{q}}'$ in $T_{\tau, \mu, \eta}(\Sigma)$. Therefore, from inequality \eqref{ineq_eta1} we conclude that $ (x'_{\tau},x'_{\textbf{q}})\in \mathcal{R}$ and condition (ii) in \Cref{def:IOS} holds. Thus, $ T_{\tau}(\Sigma)\preceq^{(\epsilon_u, \epsilon_y)}_{IOS} T_{\tau,\mu,\eta}(\Sigma)$.
\newline  \indent Along similar lines we can prove that $T_{\tau,\mu,\eta}(\Sigma)\preceq^{(\epsilon_u, \epsilon_y)}_{IOAS} T_{\tau}(\Sigma)$. The key is to define $\mathcal{R} \subseteq X_q \times X_\tau$ such that $(x_q, x_\tau)\in \mathcal{R}$ if $\|x_q-x_\tau\|=0$ and notice that for every $x_q \in X_q$, we can choose $x_\tau = x_q, x_\tau \in X_\tau$ which satisfies condition (i) of definition 4 (i.e., $\|x_{\tau}-x_{\textbf{q}} \|=0<\epsilon_y$). This is possible because $X_{\textbf{q}} \subseteq X_{\tau}$. 
\newline For every $u_q \in U_q$, choose $u_{\tau}=u_{\textbf{q}}, u_\tau\in U_\tau$ (this satisfies $\|u_{\tau}-u_{\textbf{q}}\|=0<\epsilon_u$). Consider the unique transition $x_{\tau} \xrightarrow [\quad \tau \quad]{u_{\tau}} x_{\tau}' = \xi(\tau, x_{\tau}, u_{\tau}) \in \textbf{Post}_{u_{\tau}}(x_{\tau}) $. The distance between $x_{\tau}'$ and $\xi(\tau, x_{\textbf{q}}, u_{\textbf{q}})$ can be bounded using the incrementally forward complete property of $\Sigma$, i.e.,
\begin{equation}\label{ioas_ineq11}
\|x_{\tau}'-\xi(\tau, x_{\textbf{q}},u_{\textbf{q}})\|\leq \alpha_1(0,\tau)+\alpha_2(0, \tau)
\end{equation}
Proceeding same as the proof of $ T_{\tau}(\Sigma)\preceq^{(\epsilon_u, \epsilon_y)}_{IOS} T_{\tau,\mu,\eta}(\Sigma)$, we can show that for every $x_\tau' \in \textbf{Post}_{u_{\tau}}(x_{\tau})$ there exists $x_q' \in \textbf{Post}_{u_{q}}(x_{q})$ such that $(x_q',x_\tau'\in\mathcal{R})$. This is condition (ii) of \Cref{def:IOAS}. Thus, $T_{\tau,\mu,\eta}(\Sigma)\preceq^{(\epsilon_u, \epsilon_y)}_{IOAS} T_{\tau}(\Sigma)$.
}
\end{proof}
\subsection*{Proof of \Cref{result2}}
\begin{proof}
At any discrete time k, the state, input and output of \(T_{\tau,\mu,\eta}\) in \Cref{sampled-quantized system different output} are $x_q$, $u_q$ and $h(x_q,u_q)$. Since the original continuous time system $\Sigma$ is QSR dissipative, inequality (\ref{QSRcontinuous}) holds for any input and all \(t_1>t_0>0\). Therefore, it will be valid even if we substitute \(t_0=k\tau,\ t_1= (k+1)\tau\) and \(u(t) = u_q \in U_{q}\subset U\) for $k\tau\leq t\leq(k+1)\tau$ where \(\tau\) is the sampling time and $x(t_0)=x_q \in X_q \subset X$. Under these conditions, system output of $\Sigma$ is $h(x(t),u_q)$. For simplicity of notation we represent \(h(x_q,u_q)\) as \(y(k\tau)\) and \(h(x(t),u_q)\) as \(y(t)\).
Dissipativity inequality becomes,
\begin{align}
\label{eq3}
\int_{k\tau}^{(k+1)\tau} (y(t)^T Q y(t) +  y(t)^T S u_q+  u_q^T R u_q)dt \geq V(\xi( \tau, x_q, u_q))-V(x_q)
\end{align}
 


\noindent\textbf{Bounds for $\mathbf{\int_{k\tau}^{(k+1)\tau} y(t)^T Q y(t) \ dt}$} 
\small
\begin{align}\label{a}
&\left|\int_{k\tau}^{(k+1)\tau} (y(t)^T Q y(t) - y(k\tau)^T Q y(k\tau)) \ dt \right|
 = \left| \int_{k\tau}^{(k+1)\tau} \int_{k\tau}^t \frac{d}{ds}(y^T(s)Qy(s))ds\ dt \right| \nonumber\\ 
& \leq 2 \int_{k\tau}^{(k+1)\tau} \int_{k\tau}^t \|Q^Ty(s)\|_2 \|\dot{y}(s)\|_2 ds\ dt
 \leq 2 \int_{k\tau}^{(k+1)\tau} \int_{k\tau}^\tau \|Q^Ty(s)\|_2 \|\dot{y}(s)\|_2 ds\ dt \nonumber \\
& \leq \tau \|Q\|_2 (\int_{k\tau}^{(k+1)\tau}\hspace{-2mm} y(t)^T y(t) \ dt + \gamma^2{\tau} u_q^T u_q)
\end{align} 
\vspace{-3mm}
\begin{align}
\label{c}
\left| \int_{k\tau}^{(k+1)\tau} \hspace{-4mm}(y(t)^T\hspace{-1mm}y(t) -  y(k\tau)^T \hspace{-1mm}y(k\tau)) \ dt \right| 
& \leq \hspace{-1mm} \int_{k\tau}^{(k+1)\tau} \hspace{-4mm} (\|y(t) - y(k\tau)\|_2^2 + 2 \|y(k\tau)\|_2\|y(t)-y(k\tau)\|_2) dt\nonumber
\end{align}
\normalsize
Using \(\|y(t) - y(k\tau)\|_2\leq \sqrt{\tau}\sqrt{\int_{k\tau}^{(k+1)\tau}  \|y(s)\|^2_2ds}\leq\tau \gamma \|u_q\|_2\) in the above equation, 
\small
$$\int_{k\tau}^{(k+1)\tau} y(t)^T y(t) \ dt \leq (\tau^2\gamma^2+\tau\gamma)  \int_{k\tau}^{(k+1)\tau} u_q^Tu_q\ dt + (\tau\gamma+1) \int_{k\tau}^{(k+1)\tau} y(k\tau)^T y(k\tau)\ dt.$$
\normalsize This can be used in (\ref{a}) to compute bounds on the first term of (\ref{eq3}), \small
\begin{multline}
\label{d}
 \int_{k\tau}^{(k+1)\tau} y(t)^T Q y(t)\ dt
\leq   y(k\tau)^T (\tau Q+\tau^2\|Q\|_2(\tau\gamma+1)\mathbf{I})y(k\tau)\\
+  u_q^T (\tau^3\|Q\|_2(\tau\gamma^2+\gamma)+\tau^2\gamma^2\|Q\|_2)\mathbf{I}u_q
\end{multline}
\normalsize
\textbf{Bounds for $\mathbf{ \int_{k\tau}^{(k+1)\tau} y(t)^T S u_q\ dt}$}
\small
\begin{flalign}
\label{ee}
&\left| \int_{k\tau}^{(k+1)\tau} y(t)^T S u_q\ dt -  \int_{k\tau}^{(k+1)\tau} y(k\tau)^TSu_q\ dt\right| \nonumber\\
&\leq \|Su_q\|_2\int_{k\tau}^{(k+1)\tau}\int_{k\tau}^t\|\dot{y}(s)\|_2dsdt
 \leq \tau\sqrt{\tau} \|Su_q\|_2 \sqrt{\int_0^\tau\|\dot{y}(s)\|^2_2ds} 
\leq  \int_{k\tau}^{(k+1)\tau} \hspace{-2mm}u_q^T \tau \gamma \|S\|_2 u_q \ dt \nonumber \end{flalign}
\begin{equation}
\Rightarrow  \int_{k\tau}^{(k+1)\tau} y(t)^T S u_q\ dt \leq  u_q^T \tau^2 \gamma \|S\|_2 u_q +  {\tau} y(k\tau)^T S u_q
\end{equation}
\normalsize
\textbf{Bounds for $V(\xi( \tau, x_q, u_q))$: }
Now we consider a  transition $x_{\textbf{q}} \xrightarrow [\quad \tau \quad]{u_{\textbf{q}}} x_{\textbf{q}}'$ in  $ T_{\tau, \mu, \eta}(\Sigma)$ and by Definition of $ T_{\tau, \mu, \eta}(\Sigma)$ we have  $\|\xi(\tau, x_{\textbf{q}}, u_{\textbf{q}})-x_{\textbf{q}}'\|\leq \alpha_1(\theta_1, \tau)+\alpha_2(\theta_2, \tau)+\eta/2$. For Lipschitz continuous storage functions
\begin{equation}\label{inq7}
V(\xi(\tau, x_{\textbf{q}}, u_{\textbf{q}}))\geq \hspace{-1mm} V(x_{\textbf{q}}')-L(\|x_{\textbf{q}}'-\xi(\tau, x_{\textbf{q}}, u_{\textbf{q}})\|) \geq V(x_{\textbf{q}}')-L(\alpha_1(\theta_1, \tau)+\alpha_2(\theta_2, \tau)+\eta/2) \end{equation}
\normalsize
Using (\ref{d}), (\ref{ee}) and (\ref{inq7}) in (\ref{eq3}) gives the result.
\end{proof}
\subsection*{Proof of Theorem \Cref{main_corr2}}
\begin{proof}
\highlighttext{We first show that  $ {T}_{\tau}(\Sigma)\preceq^{(\epsilon_u, \epsilon_y)}_{IOS} {T}_{\tau,\mu,\eta}(\Sigma)$.  Consider a relation \(\mathcal{R}\subset X_\tau\times X_q\) such that $(x_\tau,x_q)\in\mathcal{R}$ if $\|x_{\tau}-x_{\textbf{q}}\| \leq \eta/2$. Note that for any $x_{\tau} \in X_{\tau}$ and $u_{\tau} \in U_{\tau}$, there always exist $x_q \in X_{\textbf{q}}$ and $u_q \in U_{\textbf{q}}$ such that $\|x_{\tau}-x_q\| \leq \eta/2$ and $\|u_{\tau}-u_q\| \leq \mu/2 \leq \epsilon_u$. This is possible because of the specific quantization which allows $x_{\tau}$ to be within $\eta/2$ radius of $x_{q}$ and $u_{\tau}$ to be within $\mu/2$ radius of $u_q$.
\newline \noindent From the definitions of output functions we have, 
\begin{align}
\|h_{\tau}(x_{\tau},u_{\tau})-h_q(x_q,u_q)\|&=\|h(x_\tau,u_\tau)-h(x_q,u_q)+h(x_q,u_q)-h_q(x_q,u_q)\| \nonumber\\
& \leq K_1\|x_\tau - x_{q}\| + K_2\|u_\tau - u_{q}\| + \|h(x_q,u_q)-h_q(x_q,u_q)\| \nonumber \\
& \leq K_1\eta/2+K_2\mu/2 +\mu/2 \leq \epsilon_y
\end{align}
Hence condition (i) of \Cref{def:IOS} is satisfied. Along the lines of proof of \Cref{main_corr1}, it can be shown that condition (ii) of \Cref{def:IOS} is also satisfied, proving $ {T}_{\tau}(\Sigma)\preceq^{(\epsilon_u, \epsilon_y)}_{IOS} {T}_{\tau,\mu,\eta}(\Sigma)$.
\newline \indent Similar to the steps described above and proof of \Cref{main_corr1}, we can show that $T_{\tau,\eta,\mu}(\Sigma)\preceq^{(\epsilon_u, \epsilon_y)}_{IOAS} T_{\tau}(\Sigma)$.}
\end{proof}
\subsection*{Proof of Theorem \Cref{result3}}
\begin{proof}
At any discrete time k, the state, input and output of \(T_{\tau,\mu,\eta}\) in \Cref{sampled-quantized system same output} are $x_q$, $u_q$ and $h_q(x_q,u_q)$. Since the original continuous time system $\Sigma$ is QSR dissipative, inequality (\ref{QSRcontinuous}) holds for any input and all \(t_1>t_0>0\). Therefore, it will be valid even if we substitute \(t_0=k\tau,\ t_1= (k+1)\tau\) and \(u(t) = u_q \in U_{q}\subset U\) for $k\tau\leq t\leq(k+1)\tau$ where \(\tau\) is the sampling time and $x(t_0)=x_q \in X_q \subset X$. 
Under these conditions, system output of $\Sigma$ is $h(x(t),u_q)$.

For simplicity of notation we represent $h(x_q,u_q)=y(k\tau)$, \(h(x(t),u_q)=y(t)\), \(h_q(x_q,u_q)=\hat{y}(k\tau)\) and \highlighttext{\(\Delta y = y(k\tau)- \hat{y}(k\tau)\)}. Dissipativity inequality becomes,
\begin{align}\label{eq3thm3}
\int_{k\tau}^{(k+1)\tau} (y(t)^T Q y(t) +  y(t)^T S u_q+  u_q^T R u_q)dt \geq V(\xi( \tau, x_q, u_q))-V(x_q)
\end{align}

	
	In order to prove the dissipativity of the approximate input-output similar system \({T}_{\tau,\mu,\eta}\), we find bounds on each term of equation (\ref{eq3thm3}).
	
	\noindent\textbf{Bounds for $\mathbf{\int_{k\tau}^{(k+1)\tau} y(t)^T Q y(t) dt}$ :} From proof of \Cref{result2}
    \small
	\begin{align}
	&\int_{k\tau}^{(k+1)\tau} y(t)^T Qy(t)dt \nonumber \\
    &\leq y(k\tau)^T (Q\tau+\tau^2\|Q\|_2(\tau\gamma+1)\mathbf{I})y(k\tau) +  u_q^T (\tau^3\|Q\|_2(\tau\gamma^2+\gamma)+\tau^2\gamma^2\|Q\|_2)\mathbf{I}u_q\nonumber\\
	\label{eq4thm3}
	 & \leq \hat{y}(k\tau)^T (Q\tau+\tau^2\|Q\|_2(\tau\gamma+1)\mathbf{I})\hat{y}(k\tau) + 2\hat{y}(k\tau)^T(Q\tau+\tau^2\|Q\|_2(\tau\gamma+1)\mathbf{I})\Delta y \nonumber\\
	\qquad &\qquad + \Delta y^T (Q\tau+\tau^2\|Q\|_2(\tau\gamma+1)\mathbf{I})\Delta y +  u_q^T (\tau^3\|Q\|_2(\tau\gamma^2+\gamma)+\tau^2\gamma^2\|Q\|_2)\mathbf{I}u_q.
	\end{align}
	\begin{align}\label{result3eq1}
	2 \hat{y}(k\tau)^T (Q\tau+\tau^2\|Q\|_2(\tau\gamma+1)\mathbf{I})\Delta y
	 &\leq \tau\Delta y^T \Delta y + \tau\|Q+\tau\|Q\|_2(\tau\gamma+1)\mathbf{I}\|^2_2 \hat{y}(k\tau)^T \hat{y}(k\tau) \nonumber \\
	&\leq \tau m \frac{\mu^2}{4} + \tau \|Q+\tau\|Q\|_2(\tau\gamma+1)\mathbf{I}\|_2^2 \hat{y}(k\tau)^T \hat{y}(k\tau) \\
\label{result3eq2}
   \Delta y^T (Q\tau+\tau^2\|Q\|_2(\tau\gamma+1)\mathbf{I}) \Delta y
    & \leq \bar{\lambda}(Q\tau+\tau^2\|Q\|_2(\tau\gamma+1)\mathbf{I}) \Delta y^T \Delta y \nonumber\\
& \leq \frac{1}{4}m \tau \mu^2 \bar{\lambda}(Q+\tau\|Q\|_2(\tau\gamma+1)\mathbf{I})
	\end{align}
	\normalsize
	where $\bar{\lambda}(\cdot)$ is the largest eigen value of matrix under consideration. Using (\ref{result3eq1}) and (\ref{result3eq2}) in (\ref{eq4thm3}),
    \small
	\begin{align}
	\label{dthm3}
	\int_{k\tau}^{(k+1)\tau} \hspace{-3mm} y(t)^T Qy(t) dt &\leq  \hat{y}(k\tau)^T (Q+\tau\|Q\|_2(\tau\gamma+1)\mathbf{I}  + (\|Q+\tau\|Q\|_2(\tau\gamma+1)\mathbf{I}\|^2_2)\mathbf{I})\tau\hat{y}(k\tau) \nonumber\\
	& \hspace{-15mm} + u_q^T (\tau^3\|Q\|_2(\tau\gamma^2+\gamma)+\tau^2\gamma^2\|Q\|_2)\mathbf{I}u_q + \frac{1}{4} m \tau \mu^2(1+\bar{\lambda}(Q+\tau\|Q\|_2(\tau\gamma+1)\mathbf{I})).
	\end{align}
\normalsize	
	\textbf{Bounds for $\mathbf{2 \int_{k\tau}^{(k+1)\tau} y(t)^T S u_q \ dt}$}
	\small
    \begin{align}
	2\int_{k\tau}^{(k+1)\tau} \hspace{-3mm} (y(t)^T S u_q \ dt - \hat{y}(0)^T S u_q)dt
     	& \hspace{-1mm}\leq \int_{k\tau}^{(k+1)\tau} \hspace{-3mm} ((y(t) - \hat{y}(k\tau))^T \hspace{-1mm} (y(t) - \hat{y}(k\tau)) + (S u_q)^T \hspace{-1mm} (Su_q))dt \nonumber \\
	& \leq \int_{k\tau}^{(k+1)\tau} \hspace{-2mm} (y(t) - \hat{y}(k\tau))^T(y(t) - \hat{y}(k\tau))dt + \tau\|S\|^2_2 u_q^Tu_q \label{result3eq3}
    \end{align}
    \begin{align}
	 &\left| \int_{k\tau}^{(k+1)\tau} (y(t) - \hat{y}(k\tau))^T(y(t) - \hat{y}(k\tau))dt \right|\nonumber\\	
	& \leq \int_{k\tau}^{(k+1)\tau} \|y(t)-\hat{y}(k\tau)\|_2^2 dt 
\leq \int_{k\tau}^{(k+1)\tau}(\|y(t) - y(k\tau)\|_2 + \|y(k\tau)-\hat{y}(k\tau)\|_2)^2dt \nonumber \\ %
&\leq \tau(\sqrt{\gamma^2\int_{k\tau}^{(k+1)\tau} \hspace{-3mm} \|u_q\|^2_2dt} + \sqrt{m}\mu/2)^2 
	\leq (\gamma^2\tau^2 + \gamma\tau\sqrt{m\tau}\mu) u_q^Tu_q + \frac{\mu\tau (m\mu+\gamma\sqrt{m\tau})}{4} \label{result3eq4}. \nonumber
	\end{align}
    \normalsize
	Using this in (\ref{result3eq3}),
    \small
	\begin{align}
	\label{ethm3}
	2\int_{k\tau}^{(k+1)\tau} \hspace{-3mm} y(t)^T \hspace{-1mm} S u_q \ dt & \leq 2\tau \hat{y}(k\tau)^T \hspace{-1mm} S u_q + (\gamma^2\tau^2 + \gamma\tau\sqrt{m\tau}\mu + \tau\|S\|^2_2) u_q^Tu_q + \frac{\mu\tau(m\mu+\gamma\sqrt{m\tau})}{4}.
	\end{align}
    \normalsize
    Similar to the proof of \Cref{result1}, $V(\xi( \tau, x_q, u_q))$ can also be bounded.
	We can use this bound on $V(\xi( \tau, x_q, u_q))$, (\ref{dthm3}) and (\ref{ethm3}) to bound the terms in (\ref{eq3thm3}) and rearrange the resulting equation to get (\ref{thm3eq1}).
	\end{proof}
\subsection*{Proof of Theorem \ref{passive_composition}}
\begin{proof}
Output of $T_1$ is $h_1({x}_1,u_1)$ and we consider two possible outputs of $T_{12}=T_2 \times_{\mathcal{F}}^{(\epsilon_u,\epsilon_y)}T_1$. From the definition of approximate feedback composition and Proposition \ref{prop1},  possible relations between $T_{12}$ and $T_1$  are given by 
\begin{align*}
\textbf{Case 1:}& \quad h_{12}(x_1, x_2,u_1,u_2) = \frac{1}{2}(h_1(x_1,u_1)+ h_2(x_2,u_2)) &\Rightarrow T_{12}\preceq_{IOS}^{(\epsilon_u,\epsilon_y/2)} T_1 \\
\textbf{Case 2:} & \quad h_{12}(x_1, x_2,u_1,u_2) = h_2(x_2,u_2)  &\Rightarrow T_{12}\preceq_{IOS}^{(\epsilon_u,\epsilon_y)} T_1.
\end{align*}
Using \Cref{result1} for both these cases gives the result.
\end{proof}
\end{document}